\documentclass[12pt]{article}
\usepackage{amssymb,amsmath,amsthm,amsxtra,amsfonts}
\usepackage{geometry}
\usepackage{aliascnt}
\usepackage{color}
\usepackage[usenames,dvipsnames,table]{xcolor}
\usepackage{subcaption}
\usepackage{bbm}
\usepackage{dsfont}
\usepackage{graphicx}
\usepackage{mathrsfs}
\usepackage{tikz}
\usetikzlibrary{plotmarks,arrows,calc,patterns}
\usepackage{pgfplots}
\usepackage[round, longnamesfirst]{natbib}
\usepackage[colorlinks=true,citecolor=Blue,urlcolor=Blue,linkcolor=Blue]{hyperref}
\usepackage{threeparttable}
\usepackage{booktabs}

\theoremstyle{theorem}
\newtheorem{theorem}{Theorem}[section]

\newtheorem{lemma}{Lemma}[section]

\newtheorem{example}{Example}[section]

\newtheorem{assumption}{Assumption}[section]
\newtheorem{definition}{Definition}[section]

\newcommand{\R}{\ensuremath{\mathbf{R}}}
\newcommand{\Z}{\ensuremath{\mathbf{Z}}}

\newcommand{\E}{\ensuremath{\mathds{E}}}
\newcommand{\one}{\ensuremath{\mathds{1}}}
\newcommand{\dx}{\ensuremath{\mathrm{d}}}

\DeclareMathOperator*{\argmin}{arg\,min}

\begin{document}

\title{High-dimensional mixed-frequency IV regression\footnote{First draft: October 2014. This paper is a substantially revisited Chapter 2 of my Ph.D. thesis. I'm deeply indebted to my advisor Jean-Pierre Florens and other members of my Ph.D. committee: Eric Gautier, Ingrid van Keilegom, and Timothy Christensen for helpful discussions and suggestions. This paper was presented at “Conference on Inverse Problems in Econometrics” at Northwestern University, "ENTER exchange seminar" at Tilburg University, "48èmes Journées de Statistique de la SFdS" in Montpellier, "3rd ISNPS Conference" in Avignon, and "Recent Advances in Econometrics" conference at TSE. I'm grateful to all participants for interesting discussions, comments, and suggestions, especially to Christoph Breunig, Federico Bugni, Samuele Centorrino, Christophe Gaillac, Eric Gautier, Joel Horowitz, Pascal Lavergne, Robert Lieli, Valentin Patilea, Jeff Racine, Mario Rothfelder, Anna Simoni, and Daniel Wilhelm. I would also like to thank Bruno Biais, Sophie Moinas, and Aleksandra Babii for helpful conversations.}}
\author{
	Andrii Babii\footnote{University of North Carolina at Chapel Hill - Gardner Hall, CB 3305
		Chapel Hill, NC 27599-3305. Email: \href{mailto:babii.andrii@gmail.com}{babii.andrii@gmail.com}.} \\
	\textit{\normalsize UNC Chapel Hill}
}
\date{\today}
\maketitle

\begin{abstract}
	This paper introduces a high-dimensional linear IV regression for the data sampled at mixed frequencies. We show that the high-dimensional slope parameter of a high-frequency covariate can be identified and accurately estimated leveraging on a low-frequency instrumental variable. The distinguishing feature of the model is that it allows handing high-dimensional datasets without imposing the approximate sparsity restrictions. We propose a Tikhonov-regularized estimator and derive the convergence rate of its mean-integrated squared error for time series data. The estimator has a closed-form expression that is easy to compute and demonstrates excellent performance in our Monte Carlo experiments. We estimate the real-time price elasticity of supply on the Australian electricity spot market. Our estimates suggest that the supply is relatively inelastic and that its elasticity is heterogeneous throughout the day.
\end{abstract}

\noindent {\footnotesize \textbf{Keywords}: high-dimensional IV regression, mixed-frequency data, identification, Tikhonov regularization, continuum of moment conditions, real-time price elasticities. \\
\noindent \textbf{JEL Classifications}: C14, C22, C26, C58}

\thispagestyle{empty}


\setcounter{page}{1}

\section{Introduction}
The technological progress over the past decades has made it possible to generate, to collect, and to store new intraday high-frequency time series datasets that are widely available along with the "old" low-frequency data. Indeed, the economic activity occurs in real time and the economic and financial transactions are frequently recorded instantaneously, while the traditional time series data are available at a quarterly, monthly, or sometimes daily frequencies. Ignoring the high-frequency nature of the data leads to the loss of the information through the temporal aggregation and makes it impossible to quantify the economic activity in real time. At the same time, combining the low and the high-frequency datasets allows obtaining more refined measures of the economic activity that can be used subsequently to inform market participants and to guide policies.

In this paper, we introduce a novel high-dimensional mixed-frequency instrumental variable (IV) regression suitable for the datasets recorded at different frequencies. The model connects a low-frequency dependent variable to endogenous covariates sampled from a continuous-time stochastic process. Alternatively, the regressor might be sampled from a continuous-space stochastic process encountered in the spatial data analysis or any other stochastic process indexed by the continuum. This leads to the high-dimensional IV regression with a large number of endogenous regressors.

The high-dimensional mixed-frequency IV regression features several remarkable properties. First, we show that it is possible to identify and to estimate accurately the high-dimensional slope parameter leveraging on a \textit{low-frequency} instrumental variable. In contrast, the point identification in the (high-dimensional) linear IV regression typically relies on the order condition postulating that the number of instrumental variables should be at least as large as the number of endogenous regressors. Second, the mixed frequency IV regression can handle arbitrary large number of endogenous covariates relatively to the sample size without relying on approximate sparsity condition and restrictive tail conditions. Such a remarkable property is possible due to the continuous-time structure of the regressor and the slope parameter. Continuous-time structures is one of the "blessings of dimensionality" according to \cite{donoho2000high}.\footnote{The continuum modeling, the concentration of measure phenomenon, and the dimension asymptotics are the three "blessings of dimensionality", according to \cite{donoho2000high}.} These properties distinguish our model from the ridge IV regression, cf., \cite{carrasco2012regularization} or the high-dimensional IV regression of \cite{belloni2012sparse}.

The high-dimensional mixed-frequency IV regression is an example of ill-posed inverse problem in the sense that the map from the distribution of the data to the slope parameter is not continuous. As a result, we need to introduce some amount of regularization\footnote{The concept of regularization originates from the mathematical literature on ill-posed inverse problems, cf., \cite{tikhonov1943stability} and \cite{tikhonov1963regularization}, see \cite{carrasco2007linear} for a review and further references in econometrics.} to smooth out the discontinuities and to obtain a consistent estimator. In this paper, we focus on the Tikhonov regularization and establish its statistical properties with weakly dependent data. The estimation accuracy of the continuous-time slope parameter depends both on its regularity as well as on the regularity of a certain integral operator.

Our empirical application extends the classical IV estimation of the supply and the demand equations, cf., \cite{wright1928tariff}, to the real-time spot markets. We collect a new dataset using publicly available data and estimate the real-time price elasticity of supply in the Australian electricity spot market. To that end, we leverage on the daily temperature as an instrumental variable that shifts the demand curve and is exogenous for supply shocks. The temperature is a valid instrumental variable since the electricity demand increases in hot and cold times due to cooling and heating needs. Our empirical results reveal that while the supply of the electricity is relatively inelastic, its elasticity is heterogeneous across the day, peaking around 6 pm and dropping subsequently to its lowest value around 4 am.

\paragraph{Contribution and related literature.} Our paper connects several strands of the literature. First, following \cite{ghysels2004midas}, \cite{ghysels2007midas}, and \cite{andreou2010regression}, there is an increasing interest in using datasets sampled at different frequencies in the empirical practice. Most of this literature, with a notable exception for \cite{ghysels2010midas} and \cite{khalaf2017dynamic}, is largely focused on the forecasting problem with mixed-frequency data and does not consider the structural econometric modeling with the instrumental variable approach. The mixed-frequency data typically lead to high-dimensional problems and the dimensionality is controlled using tightly parametrized weight functions, see also \cite{foroni2015unrestricted} for the unrestricted mixed-frequency data models. Our paper has the following features: 1) we introduce a novel instrumental variable regression suitable for the data sampled at mixed frequencies and the structural econometric modeling; 2) we do not rely on a particular parameterization of the weight function; 3) our high-frequency data are generated from the endogenous continuous-time stochastic process and we study the in-fill asymptotics.

Second, we build on insights from literature on the Tikhonov regularization of ill-posed inverse problems in econometrics, see \cite{carrasco2007linear}, \cite{gagliardini2012tikhonov}, and \cite{carrasco2013asymptotic} for comprehensive surveys, and the functional linear IV regression, see \cite{florens2015instrumental}, \cite{benatia2017functional}, and \cite{babii2020a}. In contrast to this literature, we show that it is possible to achieve identification and to estimate accurately the slope parameter relying on a \textit{single} instrumental measured at a \textit{low-frequency} only. The structure of our model is also qualitatively different and leads to the conditional expectation operator that was not previously encountered in the ill-posed inverse problems literature.

Lastly, following the influential work of \cite{belloni2011lasso}, \cite{belloni2012sparse}, and \cite{belloni2014inference}, there is an increasing interest in the estimation and inference with high-dimensional datasets in econometrics.\footnote{The literature on approximately sparse econometric models is vast, see \cite{belloni2018high} for an excellent introduction and further references.} In particular, \cite{belloni2012sparse} propose to use the LASSO to address the problem of many instruments and the nonparametric series estimation of the optimal instrument. Our mixed-frequency IV regression is qualitatively different from the above models and does not impose the approximate sparsity on the high-dimensional slope coefficients, see \cite{babii2019estimation} for a comprehensive treatment of approximately sparse mixed-frequency time series regressions. The problem of the optimal instrument is more challenging in our nonparametric setting and is left for future research, see \cite{florens2015there} for some steps in this direction.

The paper is organized as follows. In section~\ref{section:model}, we present the mixed-frequency IV regression, illustrate several economic examples, and discuss the main identification issues. In section~\ref{section:asymptotics}, we present the Tikhonov-regularized estimator and derive its statistical properties for the weakly dependent time series data. All technical details appear in the appendix. We report on a Monte Carlo study in section~\ref{section:mc} which provides further insights about the validity of asymptotic analysis in finite samples typically encountered in empirical applications. Section~\ref{sec:application} presents an empirical application to the estimation of real-time supply elasticities. Lastly,  section~\ref{section:conclusions} concludes.

\section{Mixed-frequency IV regression}\label{section:model}
The purpose of this section is to introduce the mixed-frequency IV regression and to discuss our identification and estimation strategies.

\subsection{The model}
Econometrician observes $\{(Y_t,Z_t(s_j),W_t):\; t=1,\dots,T,j=0,1,\dots,m \}$, where $Y_t\in\R$ is a low-frequency dependent variable, $Z_t(s_j)$ is a realization of a real-valued continuous-time stochastic process $Z_t = \left\{Z_t(s):\;s\in S\subset\R^d\right\}$, and $W_t\in\R^q$ is a (vector of) low-frequency instrumental variables.\footnote{If $d=1$, then $S\subset \R$ can be interpreted as a time index. More generally, if $d=2$, then $S\subset\R^2$ can be a geographical location (spatial process), and if $d=3$, then $S\subset\R^3$ can denote both the space and the time dimension (spatio-temporal process). Regardless of the dimension $d$, we always refer to $Z_t$ as a continuous-time stochastic process.} The number of high-frequency observations $m$ is left unrestricted and can (potentially) be much larger than the sample size $T$. The mixed-frequency IV regression is described as
\begin{equation*}
	Y_t = \int_ S\beta(s)Z_t(s)\dx s + U_t,\qquad \E[U_t|W_t] = 0,\qquad t=1,\dots,T.
\end{equation*}
Note that since the regressor is sampled from a high-frequency covariate $Z_t$, 

Discretizing the continuous-time equation, we obtain 
\begin{equation*}
	Y_t = \sum_{j=1}^{m}\beta(s_j)Z_t(s_j)(s_{j} - s_{j-1}) + U_t,\qquad \E[U_t|W_t] = 0,\qquad t=1,\dots,T.
\end{equation*}
It is worth stressing that the discretization of the continuous-time model leads to a consistent definition of regression slopes across different frequencies, cf., \cite{sims1971discrete} and \cite{geweke1978temporal}. In contrast, the naive discrete-time regression equation does not impose any normalization and the magnitude of the slope parameter is different across different frequencies.

The following three examples provide several empirical settings where our mixed-frequency IV regression model could be useful.
\begin{example}[Real-time price elasticities]
	Spot markets operate in real-time with commodities traded for immediate delivery. The mixed-frequency IV regression can be used to estimate the real-time elasticities of supply/demand, which is a continuous-time extension of the classical linear IV regression, cf., \cite{wright1928tariff}. In our empirical application, $Y_t$ is the quantity sold at the spot market on a day $t$ and $Z_{t}(s)$ is the equilibrium market price on a day $t$ at time $s$. The market equilibrium leads to the endogeneity problem. Using daily temperatures as a demand shifter, we can identify the real-time price elasticity $\beta$ of the electricity supply. 
\end{example}

\begin{example}[Intraday liquidity]
	In the equilibrium of a seminal \cite{kyle1985continuous} model, $Y_t$ is a daily price change of an asset $t$, $Z_t(s)$ is an order flow imbalance on a day $t$ at time $s$, and $1/\beta$ is a liquidity parameter. The liquidity parameter quantifies the sensitivity of the market price to the imbalance between the supply and the demand. Endogeneity comes from the strategic behavior of informed traders who are likely to distribute orders over time to minimize the impact on prices and the market equilibrium.
\end{example}

\begin{example}[Measurement errors]
	Classical measurement errors in the high-frequency regressor sampled from a continuous-time stochastic process also lead to the endogeneity problem. Such measurement errors are especially pronounced in the high-frequency intraday financial data contaminated by the market microstructure noise, see \cite{zhang2005tale} and \cite{hansen2006realized}.
\end{example}

\subsection{Identification}
To simplify the notation, in this section, we suppress the dependence of $(Y_t,Z_t,W_t)$ on $t$ and write $(Y,Z,W)$, which is well-justified under stationarity. The mixed-frequency IV regression becomes\footnote{Alternatively, if we start from the linear model $Y=\Phi(Z)+U$, where $\Phi:L_2(S)\to \R$ is a continuous linear functional, then by the Riesz representation theorem, we can always write $\Phi(Z) = \langle \beta,Z\rangle$ for a unique slope parameter $\beta\in L_2(S)$. Here and later, $L_2(S)$ denotes the set of real functions on $S$, square-integrable with respect to the Lebesgue measure and the natural inner product $\langle.,.\rangle$, see Appendix for more details on the notation.}
\begin{equation*}
	Y = \int_S \beta(s)Z(s)\dx s + U,\qquad\E[U|W] = 0.
\end{equation*}
The identification in the linear IV regression relies on the uncorrelatedness between the instrumental variable and the unobservables, i.e., $\E[UW]=0$, and the rank condition. The rank condition requires in turn that the number of the instrumental variable matches the dimension of the endogenous covariate. In our settings, the endogenous covariate is a high-dimensional realization of a continuous-time stochastic process, which requires in turn a high-dimensional instrumental variable. Given that the instrumental variable has to be exogenous to the system, this imposes a strong requirement on the instrumental variable.

In contrast, our identification strategy relies on the mean independence exogeneity condition, $\E[U|W]=0$. Assuming that the order of the integration can be interchanged, the exogeneity leads to
\begin{equation}\label{eq:main}
	h(w) \triangleq \E[Y|W=w] = \int_ S\beta(s)\E[Z(s)|W=w]\dx s \triangleq (L\beta)(w),
\end{equation}
where $L:L_2(S)\to L_2(W)$ is an integral operator mapping the unknown slope parameter $\beta$ to the conditional mean function $h$.\footnote{For a random variable $W$, we denote $L_2(W) = \{f:\; \E|f(W)|^2<\infty \}$ with some abuse of notation.} Eq.~\ref{eq:main} is an example of the Fredholm integral equation of type I solving, which is typically known to be ill-posed in the sense that the inverse map from $h$ to $\beta$ is discontinuous, see \cite{carrasco2007linear}.

Our identification strategy relies on the linear completeness property of the distribution of $(Z,W)$. We say that the stochastic process $Z\in L_2(S)$ is \textit{linearly complete}\footnote{The linear completeness condition is significantly weaker than the nonlinear completeness condition typically used in the nonparametric IV literature, cf., \cite{babii2017completeness}.} for $W\in\R^q$ if for all $b\in L_2(S)$ with $\E|\langle Z,b\rangle|<\infty$, we have
\begin{equation*}
	\E\left[\langle Z,b\rangle|W\right] = 0\implies b = 0.
\end{equation*}

\begin{assumption}\label{as:contiv_td}
	The stochastic process $Z$ is linearly complete for $W$.
\end{assumption}
The linear completeness is a generalization of the rank condition imposed in the finite-dimensional linear IV regression and requires that the operator $L$ is injective. Consider another injective operator $M:L_2(W)\to L_2(S)$ such that $(M\varphi)(u) = \E[\varphi(W)\Psi(u,W)]$ for some square-integrable function of the instrumental variable $\Psi$. Applying $M$ to both sides of Eq.~\ref{eq:main} leads to
\begin{equation}\label{eq:main_new}
	r(u) = \E[Y\Psi(u,W)] = \int_{ S}\beta(s)\E[Z(s)\Psi(u,W)]\dx s = (K\beta)(u),
\end{equation}
where $r=Mh$ and $K=ML$ is a new operator. It is more convenient to estimate the slope parameter $\beta$ using the continuum of moment restrictions in Eq.~\ref{eq:main_new}, since it does not involve conditional expectations, nonparametric estimation of which involves additional tuning parameters.\footnote{The problem of estimating a \textit{finite-dimensional} parameter using a continuum of moment conditions is addressed, e.g., in \cite{carrasco2000generalization} and \cite{carrasco2007efficient}.} At the same time, Eq.~\ref{eq:main_new} has the same identifying power as Eq.~\ref{eq:main} provided that the operator $M$ is injective. A large class of instrument functions $\Psi$ that ensure injectivity of $M$ is characterized in \cite{stinchcombe1998consistent}.\footnote{See also an earlier work of \cite{bierens1982consistent} who develops consistent specification tests and the work of \cite{dominguez2004consistent} and \cite{lavergne2013smooth} who develop estimators of \textit{finite-dimensional} parameters based on the Bierens-type trick.} Our default recommendation is the logistic CDF, $\Psi(u,W) = \frac{1}{1 + e^{-u^\top W}}$, which real-valued and bounded.

\section{Tikhonov regularization}\label{section:asymptotics}
In this section, we introduce the Tikhonov-regularized estimator of the slope parameter $\beta$ and study its statistical properties with time series data.

\subsection{Estimator}
Our objective is to estimate the slope parameter $\beta$ using the continuum of moment conditions in Eq.~\ref{eq:main_new}, which requires inverting the operator $K$. Note that the integral operator $K$ has the kernel function $k(s,u)\triangleq \E[Z(s)\Psi(u,W)]$, which is typically square-integrable. Consequently, the operator $K$ is compact and its generalized inverse is not continuous, see \cite{carrasco2007linear}. The operator inversion problem is amplified by the fact that $r$ and $K$ are unobserved and have to be estimated from the data. In this paper, we focus on the Tikhonov-regularized estimator of $\beta$.

Let $(Y_t,Z_t,W_t)_{t=1}^T$ be a stationary sample. The operator $K$ and the function $r$ are estimated using sample means
\begin{equation*}
\begin{aligned}
	\hat r(u) & = \frac{1}{T}\sum_{t=1}^TY_t\Psi(u,W_t),\\
	(\hat K\beta)(u) & = \int_{S}\beta(s)\hat k(s,u)\dx s,\\
	\hat k(s,u) & = \frac{1}{T}\sum_{t=1}^TZ_t(s)\Psi(u,W_t).
\end{aligned}
\end{equation*}
The Tikhonov-regularized estimator solves the following penalized least-squares problem
\begin{equation*}
	\hat{\beta} = \argmin_{b}\left\|\hat Kb - \hat r\right\|^2 + \alpha \|b\|^2,
\end{equation*}
where $\alpha>0$ is a tuning parameter controlling the amount of the regularization and $\|.\|$ is the natural norm on the relevant $L_2$ space. The estimator has a well-known closed-form expression, which resembles the expression of the finite-dimensional ridge regression estimator\footnote{It is well-known that the compact self-adjoint operator has a countable, decreasing to zero sequence of eigenvalues. Tikhonov regularization stabilizes the spectrum of the generalized inverse of the operator $\hat K^*\hat K$, replacing its eigenvalues $\frac{1}{\hat\lambda_j}$ by $\frac{1}{\alpha + \hat\lambda_j}$, see \cite{carrasco2007linear} for more details.}
\begin{equation}\label{eq:tikhonov}
	\hat\beta = (\alpha I + \hat K^* \hat K)^{-1}\hat K^*\hat r,
\end{equation}
where $\hat K^*$ is the adjoint operator to $\hat K$. To compute the adjoint operator, note that for every $\psi\in L_2$, by Fubini's theorem
\begin{equation*}
\begin{aligned}
	\langle \hat K\beta,\psi\rangle & = \int\left(\int\beta(s)\hat k(s,u)\dx s\right)\psi(u)\dx u \\
	& = \int\beta(s) \left(\int\psi(u)\hat k(s,u)\dx u\right)\dx s \\
	& = \langle\beta, \hat K^*\psi\rangle.
\end{aligned}
\end{equation*}
Therefore, the adjoint operator is
\begin{equation*}
	(\hat K^*\psi)(s) = \int\psi(s)\hat k(s,u)\dx u.
\end{equation*}

\subsection{Statistical properties}
To investigate the statistical properties of $\hat\beta$, we introduce several weak-dependence conditions on the underlying stochastic processes. The following definition generalizes the notion of the covariance stationarity to function-valued stochastic processes, see \cite{bosq2012linear} for a comprehensive introduction to the statistical theory of stochastic processes in Hilbert and Banach spaces.
\begin{definition}
	The $L_2(S)$-valued stochastic process $(X_t)_{t\in\Z}$ is covariance stationary if
	\begin{enumerate}
		\item[(i)] the second moment exists: $\sup_{t\in\Z}\E\|X_t\|^2<\infty$;
		\item[(ii)] the mean function is constant over time: $\E[X_t(s)] = \mu(s),\forall s\in S$ and $\forall t\in\Z$;
		\item[(iii)] the autocovariance function depends only on the distance between observations: $\forall s,u\in S$ and $\forall h,k\in\mathbf{Z}$
		\begin{equation*}
		\begin{aligned}
		\gamma_{h,k}(s,u) & = \E[(X_{h}(s) - \mu(s))(X_{k}(u) - \mu(u))] \\
		& = \E[(X_{|h-k|}(s) - \mu(s))(X_0(u)-\mu(u))], \\
		& \triangleq \gamma_{h-k}(s,u).
		\end{aligned}
		\end{equation*}
	\end{enumerate}
\end{definition}

We also need a notion of the absolute summability of the autocovariance function for $L_2(S)$-valued stochastic processes.

\begin{definition}
	The $L_2(S)$-valued covariance stationary process $(X_t)_{t\in\Z}$ has the absolutely summable autocovariance function $\gamma_h$ if
	\begin{equation*}
		\sum_{h\in \Z}\|\gamma_h\|_1 <\infty,
	\end{equation*}
	where $\|\gamma_h\|_1 = \int_S|\gamma(s,s)|\dx s$ denotes the $L_1$ norm on the diagonal of $S\times S$.
\end{definition}

The following assumption restricts the dependence structure of the process.
\begin{assumption}\label{as:data_ts}
	$\{u\mapsto Y_t\Psi(u,W_t):\; t\in\mathbf{Z}\}$ and $\{(s,u)\mapsto Z_t(s)\Psi(u,W_t):\; t\in\mathbf{Z}\}$ are covariance stationary $L_2$-valued stochastic processes with absolutely summable autocovariance functions.
\end{assumption}
Assumption (i) is a relatively mild condition and is satisfied, in particular, when $(Y_t,Z_t,W_t)_{t\in\mathbf{Z}}$ is strictly stationary. The absolute summability of autocovariances is also a relatively mild condition that is typically assumed in the time series analysis. It is worth stressing that the stationarity is imposed on \textit{entire} trajectories of the processes over $t\in\Z$. At the same time, on a fixed day $t\in\Z$, the intraday observations $Z_t(s)$ for $s\in S$ can be non-stationary.

Since the mixed-frequency IV regression model is ill-posed, we also need to quantify the degree of ill-posedness of the operator $K$ and the regularity of the slope parameter $\beta$. The following conditions serve this purpose.
\begin{assumption}\label{as:source}
	The slope parameter $\beta$ belongs to the class
	\begin{equation*}
	\mathcal{F}(\gamma,R) = \left\{b \in L_2(S):\; b = (K^*K)^{\gamma}\psi,\;\|\psi\|^2\leq R\right\}
	\end{equation*}
	for some $\gamma\in(0,1]$ and $R>0$.
\end{assumption}
To appreciate this condition, note that if $\beta = (K^*K)^{\gamma}\psi$, then $\psi = (K^*K)^{-\gamma}\beta$. Let $(\sigma_j,\beta_j,\psi_j)_{j=1}^\infty$ be the singular values decomposition of the compact linear operator $K$, see \cite{carrasco2007linear}. Then $\beta = \sum_{j=1}^\infty\langle\beta,\beta_j\rangle\beta_j$ and by the Parseval's identity
\begin{equation*}
	\|\psi\|^2 = \sum_{j=1}^\infty\frac{|\langle\beta,\beta_j\rangle|^2}{\sigma_j^{4\gamma}}.
\end{equation*}
Therefore, $\beta = (K^*K)^\gamma \psi$ and $\|\psi\|^2\leq R$ in Assumption~\ref{as:source} restrict the regularity of the slope parameter $\beta$ as measured by how fast the Fourier coefficients $(\langle\beta,\beta_j\rangle)_{j=1}^\infty$ decrease to zero relatively to the smoothing properties of the operator $K$ as measured by how fast the singular values $(\sigma_j)_{j=1}^\infty$ decrease to zero and the regularity parameter $\gamma>0$.

The following result provides statistical guarantees on the estimation accuracy for the Tikhonov-regularized estimator in the mean-integrated squared error.
\begin{theorem}\label{thm:main_tikhonov}
	Suppose that Assumptions~\ref{as:contiv_td}, \ref{as:data_ts}, and \ref{as:source} are satisfied. Then
	\begin{equation*}
		\E\left\|\hat \beta - \beta\right\|^2 \leq C\left(\frac{1}{\alpha T} + \frac{\alpha^{2\gamma} + \alpha^{2\gamma\wedge 1}}{\alpha T} + \alpha^{2\gamma}\right),
	\end{equation*}
	where the constant $C$ can be found in Eq.~\ref{eq:constant}.
\end{theorem}
Consequently, if the regularization parameter $\alpha$ tends to zero, we obtain
\begin{equation*}
	\E\|\hat \beta - \beta\|^2 = O\left(\frac{1}{\alpha T} + \alpha^{2\gamma}\right).
\end{equation*}
The two terms are balanced for $\alpha\sim T^{-\frac{1}{2\gamma + 1}}$, in which case the convergence rate of the integrated MSE is $O\left(T^{-\frac{2\gamma}{2\gamma + 1}}\right)$. The uniform inference for the Tikhonov-regularized estimator is also possible, cf., \cite{babii2020a}. Lastly, one could also consider regularization with Sobolev norm penalty and/or more general spectral regularization schemes, see \cite{carrasco2007linear}, \cite{carrasco2013asymptotic}, \cite{babii2017completeness}, and \cite{babii2020b}.

\subsection{Infill asymptotics}
So far we have assumed that the trajectory of the stochastic process $\{Z_t(s):\; t=1,\dots,T,s\in S\}$ is completely observed. In this section, we relax this requirement and investigate the case when we only observe $\{Z_t(s_j):\; t=1,\dots,T,\;j=1,\dots,m\}$, i.e., realizations of the process at discrete time points $s_j\in S,j=1,\dots,m$. For simplicity of presentation, suppose that $S=[0,1]$ and that $0=s_0\leq s_1< s_2 <\dots< s_m=1$. 

Then the operator
\begin{equation*}
\begin{aligned}
	(\hat K\phi)(u) & = \int_0^1\phi(s)\hat k(s,u)\dx s \\
	\hat k(s,u) & = \frac{1}{T}\sum_{t=1}^TZ_t(s)\Psi(u,W_t)
\end{aligned}
\end{equation*}
is not accessible in practice the continuous-time stochastic process $Z_t$ is only partially observed. Instead, we observe its discrete-time approximation for every $\phi\in C[0,1]$
\begin{equation*}
	(\hat K_m\phi)(u) = \sum_{j=1}^m\phi(s_j)\hat k(s_j,u)\delta_j
\end{equation*}
with $\delta_j = s_j - s_{j-1}$. Let $\Delta_m \triangleq \max_{1\leq j\leq m}\delta_j$ and let $\hat\beta_m$ be the solution to
\begin{equation*}
	(\alpha I + \hat K^*\hat K_m)\hat\beta_m = \hat K^*\hat r.
\end{equation*}

For the in-fill asymptotics, we need additionally the following assumption.
\begin{assumption}\label{as:discretization}
	(i) The process $Z$ has trajectories in the H\"{o}lder class $C_L^\kappa[0,1]$ for some $\kappa\in(0,1)$ and $L\in(0,\infty)$; (ii) $\sup_{w}\|\Psi(.,w)\|^2\leq \bar\Psi<\infty$; (iii) $\Delta_m = O(\alpha^{1/\kappa}/T^{1/2\kappa})$ as $\alpha\to 0$ and $T\to\infty$.
\end{assumption}
Assumption~\ref{as:discretization} (i) is satisfied, e.g., for the Brownian motion on $[0,1]$ with $\kappa<1/2$. (ii) is satisfied, e.g., for uniformly bounded instrument functions on compact intervals. (iii) imposes restrictions on the in-fill asymptotics. In the special case of the uniform spacing, it reduces to the condition $m^{-\kappa} = O(\alpha/T^2)$. In other words, the number of regressors should increase sufficiently fast. It is worth stressing that the number of regressors $m$ can be much larger than the sample size $T$ and can increase even faster than exponentially.

The following result shows that the integrated MSE can converge at the same rate as if we observed the process, cf., Theorem~\ref{thm:main_tikhonov}.
\begin{theorem}\label{thm:main_tikhonov_discretization}
	Suppose that Assumptions~\ref{as:contiv_td}, \ref{as:data_ts}, \ref{as:source}, and \ref{as:discretization} are satisfied. Then
	\begin{equation*}
	\E\left\|\hat \beta_m - \beta\right\|^2 \leq O\left(\frac{1}{\alpha T} + \alpha^{2\gamma}\right)
	\end{equation*}
	for some constants $C<\infty$.
\end{theorem}

\section{Monte Carlo experiments}\label{section:mc}
In this section, we discuss the numerical implementation of our high-dimensional mixed-frequency IV estimator and study its behavior in finite samples with Monte Carlo experiments.

We use the logistic CDF, $\Psi(u,W) = \frac{1}{1+\exp(-u^\top W)}$, as an instrument function.\footnote{This function fits our assumptions since it is uniformly bounded and real-valued, unlike some other choices, cf., \cite{bierens1982consistent} and \cite{stinchcombe1998consistent}. At the same time, we find in Monte Carlo experiments that it works significantly better than, e.g., $\Psi(s,w)=\one\{w\leq s \}$.} We rewrite Eq.~\ref{eq:tikhonov} as $\alpha \hat{\beta} + \hat{K}^*\hat{K}\hat\beta = \hat{K}^*\hat r$ and discretize it with the Riemann sum on a grid of uniformly spaced points $j/m,j=1,\dots,m$. The discretized equation is
\begin{equation*}
	\alpha\boldsymbol{\hat\beta} + \mathbf{Z}^\top\boldsymbol{\Psi}^\top\boldsymbol{\Psi} \mathbf{Z}\boldsymbol{\hat\beta}/(Tm)^2 = \mathbf{Z}^\top\boldsymbol{\Psi}^\top\boldsymbol{\Psi} \mathbf{y}/T^2m,
\end{equation*}
where $\boldsymbol{\hat\beta}= (\hat{\beta}(j/m))_{1\leq j\leq m}, \mathbf{Z} = (Z_t(j/m))_{\substack{1\leq t\leq T \\ 1\leq j\leq m}},\boldsymbol{\Psi} = (\Psi(j/m,W_t))_{\substack{1\leq j\leq m \\ 1\leq t\leq T}},\mathbf{y} = (Y_t)_{1\leq t\leq T}$ and $\mathbf{I}_T$ is a $T\times T$ identity matrix. Then we compute the estimator as
\begin{equation*}
	\boldsymbol{\hat\beta} = \left(\alpha\mathbf{I}_T + \mathbf{Z}^\top\boldsymbol{\Psi}^\top\boldsymbol{\Psi} \mathbf{Z}/(Tm)^2\right)^{-1}\mathbf{Z}^\top\boldsymbol{\Psi}^\top\boldsymbol{\Psi} \mathbf{y}/T^2m.
\end{equation*}

There are $5,000$ replications in each Monte Carlo experiment. We generate samples of $(Y_t,Z_t,W_t)_{t=1}^T$ of size $T\in\{100,500,1000\}$ as follows
\begin{equation*}
\begin{aligned}
Y_t & = \int_0^1\beta(s)Z_t(s)\dx s + U_t,\\
Z_t(s) & = k(s,W_t) + \sigma B_t(s),\qquad W_t = 0.5 + 0.7W_{t-1} + \varepsilon_t\\
k(s,w) &= \sqrt{s^2+w^2},\qquad \varepsilon_t\sim_{i.i.d.} N(0,1) \\
U_t & = 0.5\int_0^1 B_t(s)\dx s + 0.5V_t,\qquad  V_t\sim_{i.i.d.} N(0,1)
\end{aligned}
\end{equation*}
where $\{B_t(s): s\in[0,1],t=1,\dots,T \}$ are independent Brownian motion, generated independently of all other variables and initiated at i.i.d. random draws from $U(-1/2,1/2)$. The parameter $\sigma\in\{0.5,1\}$ represents the noise level. We consider two slope parameters $\beta(s)=-10\exp(s)$ and $\beta(s)=10s$ with $s\in[0,1]$. All continuous-time quantities are discretized at $200$ equidistant points.

The integrated bias, variance, and MSE are approximated by the Riemann sum on a grid of $100$ equidistant points in $[0,1]$. Table~\ref{tab:2_1} and Table~\ref{tab:2_2} present the results of our Monte Carlo experiments for two different population slope parameters. The mixed-frequency IV estimator behaves according to our asymptotic results. We can see the bias/variance trade-off -- as the regularization parameter $\alpha$ tends to zero, the bias decreases while the variance increases. The optimal choice of the regularization parameter should balance the two. The estimator performs better when the sample size increases and the noise level decreases. We can also see that the linear slope parameter is estimated more accurately. Figure~\ref{fig:2_1} and Figure~\ref{fig:2_2} summarize graphically the outcome of Monte Carlo experiments for $\alpha = 10^{-6}$. The shaded gray area represents the pointwise $95\%$ confidence interval across $5,000$ replications. Overall, the mixed-frequency IV estimator demonstrates excellent performance across different specifications.

It is worth stressing that since the stochastic $Z$ is observed at $m=200$ time points, the number of endogenous regressors exceeds the sample size when $T=100$. In this case, the conventional IV estimator does not exist. At the same time, the naive generalization of the ridge regression and the LASSO are also not appropriate in our setting. The ridge regression would typically require $m/T\to 0$, cf., \cite{carrasco2007linear}. The LASSO would require the approximate sparsity, somewhat stronger weak dependence conditions, and $m^{1/\kappa}/T^{1-1/\kappa}\to 0$, where $\kappa$ measures tails and weak dependence, cf., \cite{babii2019estimation}.

\begin{table}
	\centering
	\begin{threeparttable}
		\caption{Experiments for $\beta(s)=-10\exp(s)$.}
		\begin{tabular}{l|l|l|r|r|r}
			\toprule
			$T$ & $\sigma$ & $\alpha$ & $\text{i-Bias}^2$ & $\text{i-Var}$   & $\text{i-MSE}$ \\
			\midrule
     		100 & 0.5 & $10^{-5}$ &3.7458  &  0.6178  &  4.3636\\
			& 	  & $10^{-6}$     &0.2189  &  0.8535  &  1.0723\\
			& 	  & $10^{-7}$     &0.0847  &  1.5591  &  1.6437\\
			& 1.0 & $10^{-5}$     &4.4351  &  1.2387  &  5.6738\\
			& 	  & $10^{-6}$     &0.6539  &  2.0718  &  2.7256\\
			& 	  & $10^{-7}$ &0.2793  &  3.1782  &  3.4575\\\hline
			500 & 0.5 & $10^{-5}$ &3.1686  &  0.1257  &  3.2944\\
			& 	  & $10^{-6}$ &0.1276  &  0.1531  &  0.2807\\
			& 	  & $10^{-7}$ &0.0427  &  0.2180  &  0.2607\\
			& 1.0 & $10^{-5}$ &3.3092  &  0.2768  &  3.5860\\
			& 	  & $10^{-6}$ &0.1724  &  0.3889  &  0.5613\\
			& 	  & $10^{-7}$ &0.0524  &  0.5324  &  0.5848\\\hline
			1000& 0.5 & $10^{-5}$ &3.0921  &  0.0633  &  3.1554\\
			& 	  & $10^{-6}$ &0.1194  &  0.0750  &  0.1944\\
			& 	  & $10^{-7}$ &0.0391  &  0.0998  &  0.1389\\
			& 1.0 & $10^{-5}$ &3.1539  &  0.1407  &  3.2946\\
			& 	  & $10^{-6}$ &0.1347  &  0.1921  &  0.3267\\
			& 	  & $10^{-7}$ &0.0422  &  0.2398  &  0.2820\\
			\bottomrule
		\end{tabular}
		\begin{tablenotes}
			\small
			\item Note: results for different sample sizes $T$, noise levels $\sigma$, and regularization parameters $\alpha$.
		\end{tablenotes}
		\label{tab:2_1}
	\end{threeparttable}
\end{table}

\begin{table}
	\centering
	\begin{threeparttable}
		\caption{Experiments for $\beta(s)=10s$.}
		\begin{tabular}{l|l|l|r|r|r}
			\toprule
			$T$ & $\sigma$ & $\alpha$ & $\text{i-Bias}^2$ & $\text{i-Var}$   & $\text{i-MSE}$ \\
			\midrule
			100 & 0.5 & $10^{-5}$ &1.6466 &   0.3971 &   2.0437\\
			& 	  & $10^{-6}$     &0.2158 &   0.6547 &  0.8705\\
			& 	  & $10^{-7}$     &0.0633  &  1.2591 &   1.3225\\
			& 1.0 & $10^{-5}$     &1.8666 &   0.6451 &   2.5117\\
			& 	  & $10^{-6}$     &0.3660  &  1.1376  &  1.5036\\
			& 	  & $10^{-7}$ &0.1030  &  2.0589  &  2.1618\\ \hline
			500 & 0.5 & $10^{-5}$ &1.3965  &  0.0860  &  1.4825\\
			& 	  & $10^{-6}$ &0.1766  &  0.1188  &  0.2954\\
			& 	  & $10^{-7}$ &0.1160  &  0.1739  &  0.2900\\
			& 1.0 & $10^{-5}$ &1.4756  &  0.1609  &  1.6365\\
			& 	  & $10^{-6}$ &0.1934  &  0.2236  &  0.4170\\
			& 	  & $10^{-7}$ &0.0946  &  0.3300  &  0.4246\\ \hline
			1000& 0.5 & $10^{-5}$ &1.3717  &  0.0425  &  1.4142\\
			& 	  & $10^{-6}$ &0.1724  &  0.0591  &  0.2315\\
			& 	  & $10^{-7}$ &0.1249  &  0.0787  &  0.2035\\
			& 1.0 & $10^{-5}$ &1.4036  &  0.0818  &  1.4855\\
			& 	  & $10^{-6}$ &0.1832  &  0.1079  &  0.2912\\
			& 	  & $10^{-7}$ &0.1139  &  0.1444  &  0.2584\\
			\bottomrule
		\end{tabular}
		\begin{tablenotes}
			\small
			\item Note: results for different sample sizes $T$, noise levels $\sigma$, and regularization parameters $\alpha$.
		\end{tablenotes}
		\label{tab:2_2}
	\end{threeparttable}
\end{table}

\begin{figure}
	\begin{subfigure}[b]{0.32\textwidth}
		\includegraphics[width=\textwidth]{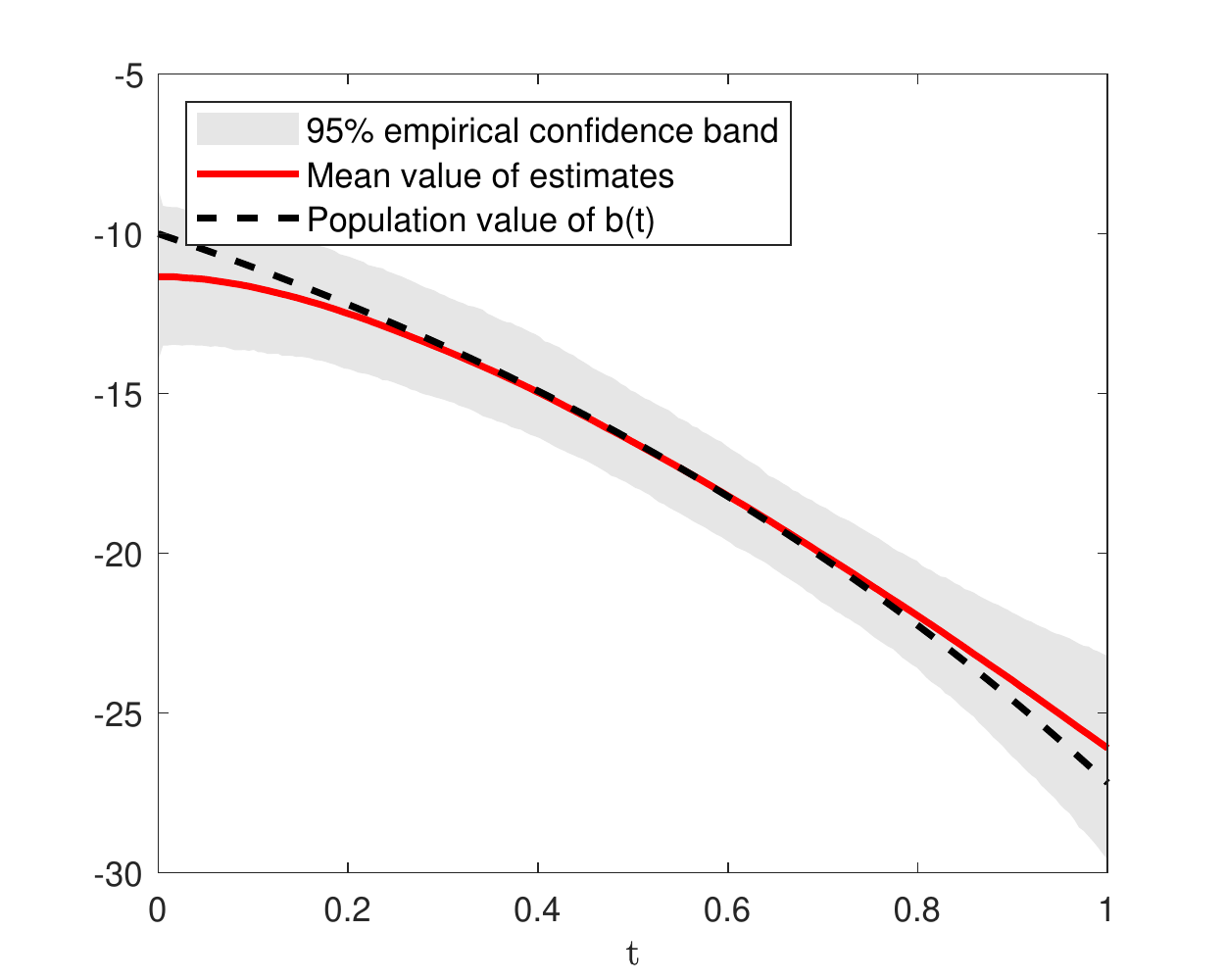}
		\caption{$\sigma=0.5,T=100$}
	\end{subfigure}
	\begin{subfigure}[b]{0.32\textwidth}
		\includegraphics[width=\textwidth]{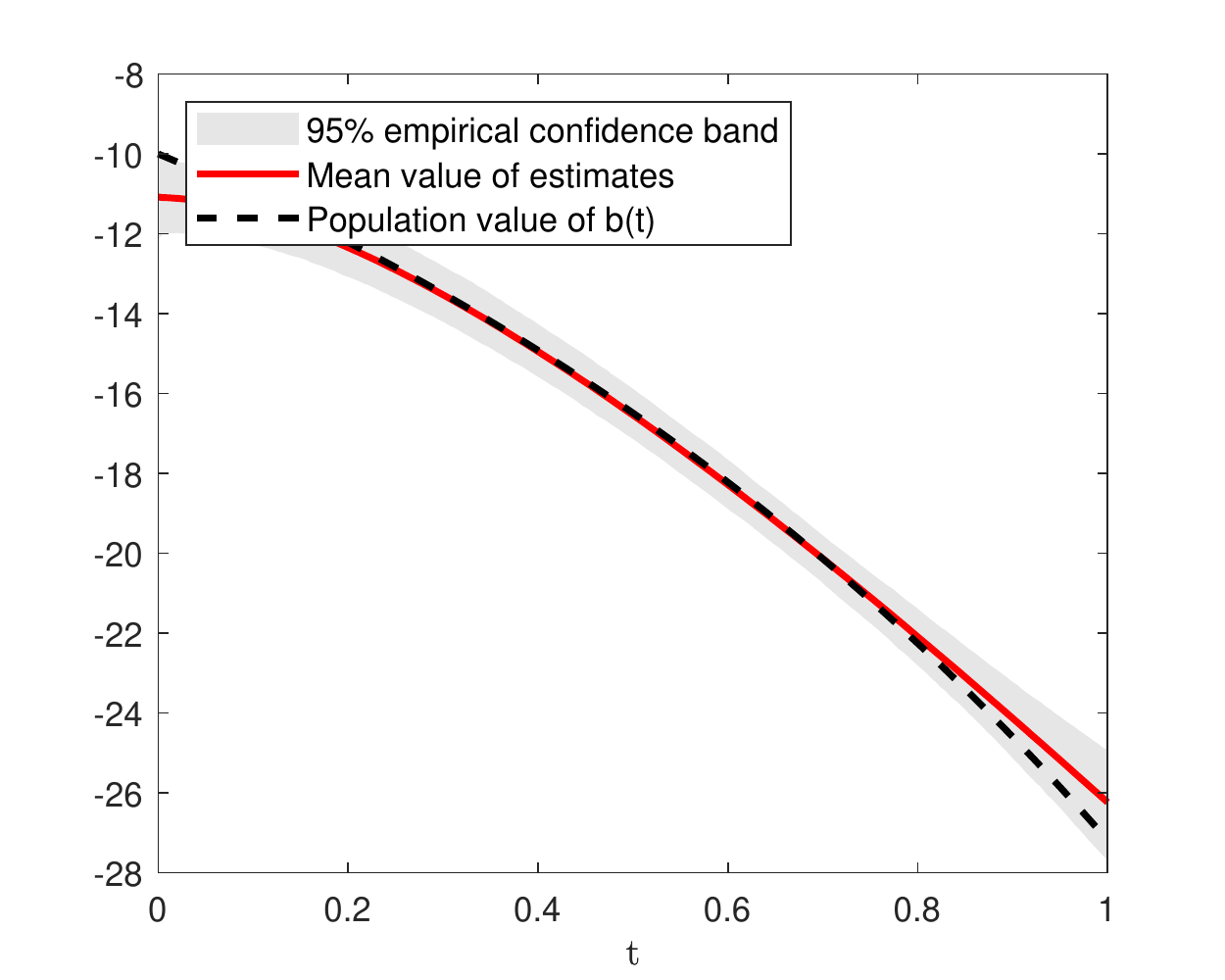}
		\caption{$\sigma=0.5,T=500$}
	\end{subfigure}
	\begin{subfigure}[b]{0.32\textwidth}
		\includegraphics[width=\textwidth]{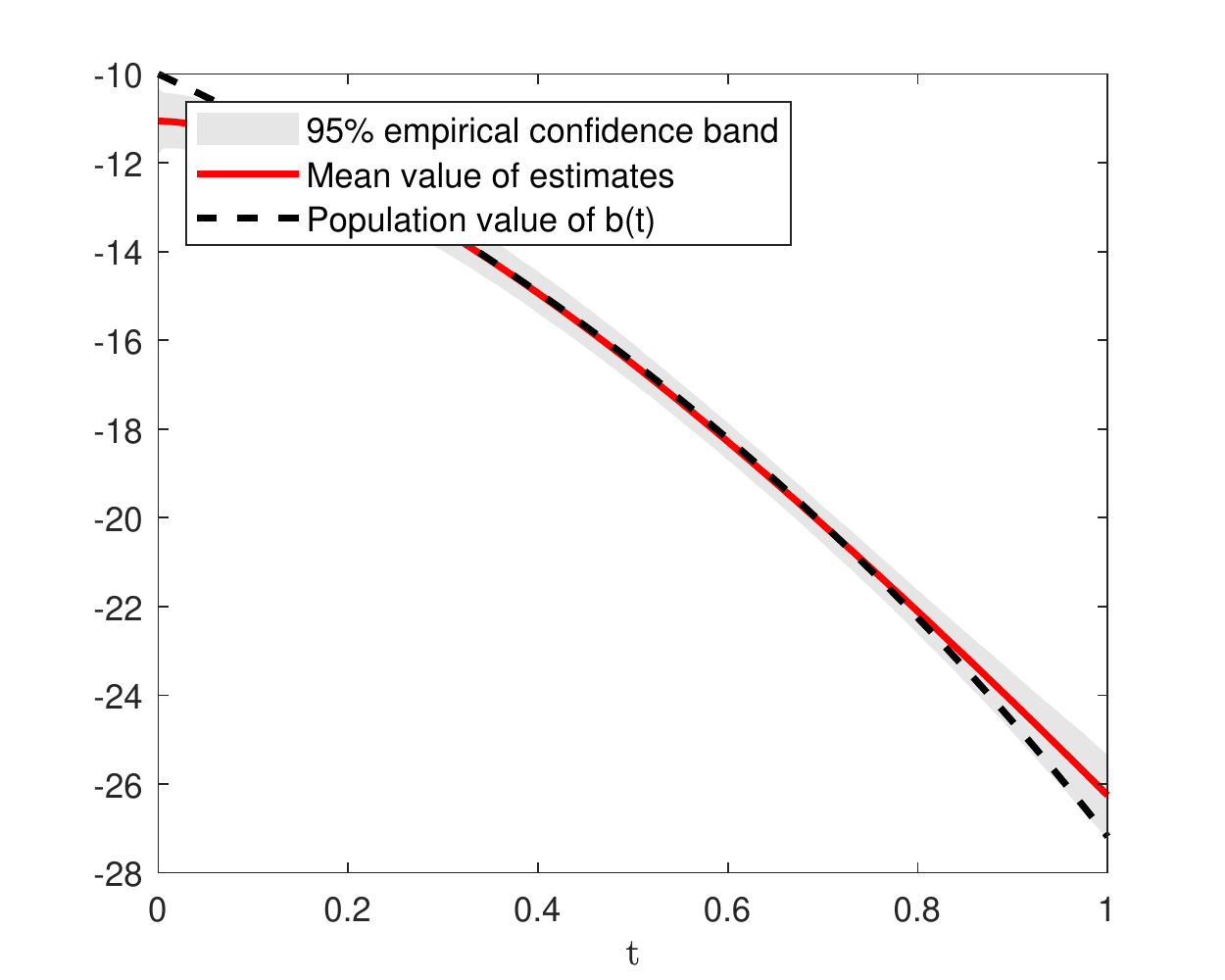}
		\caption{$\sigma=0.5,T=1000$}
	\end{subfigure}
	\begin{subfigure}[b]{0.32\textwidth}
		\includegraphics[width=\textwidth]{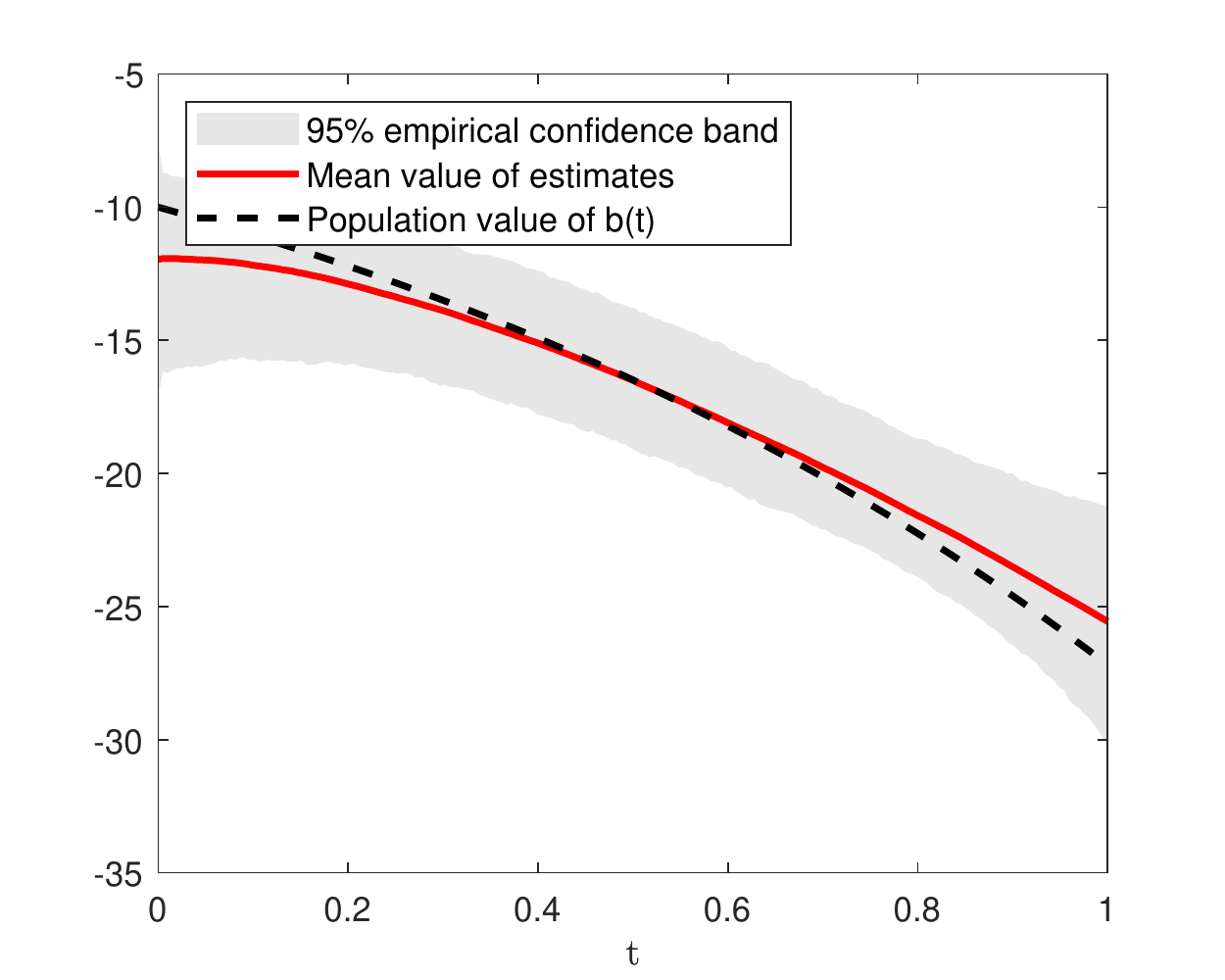}
		\caption{$\sigma=1,T=100$}
	\end{subfigure}
	\begin{subfigure}[b]{0.32\textwidth}
		\includegraphics[width=\textwidth]{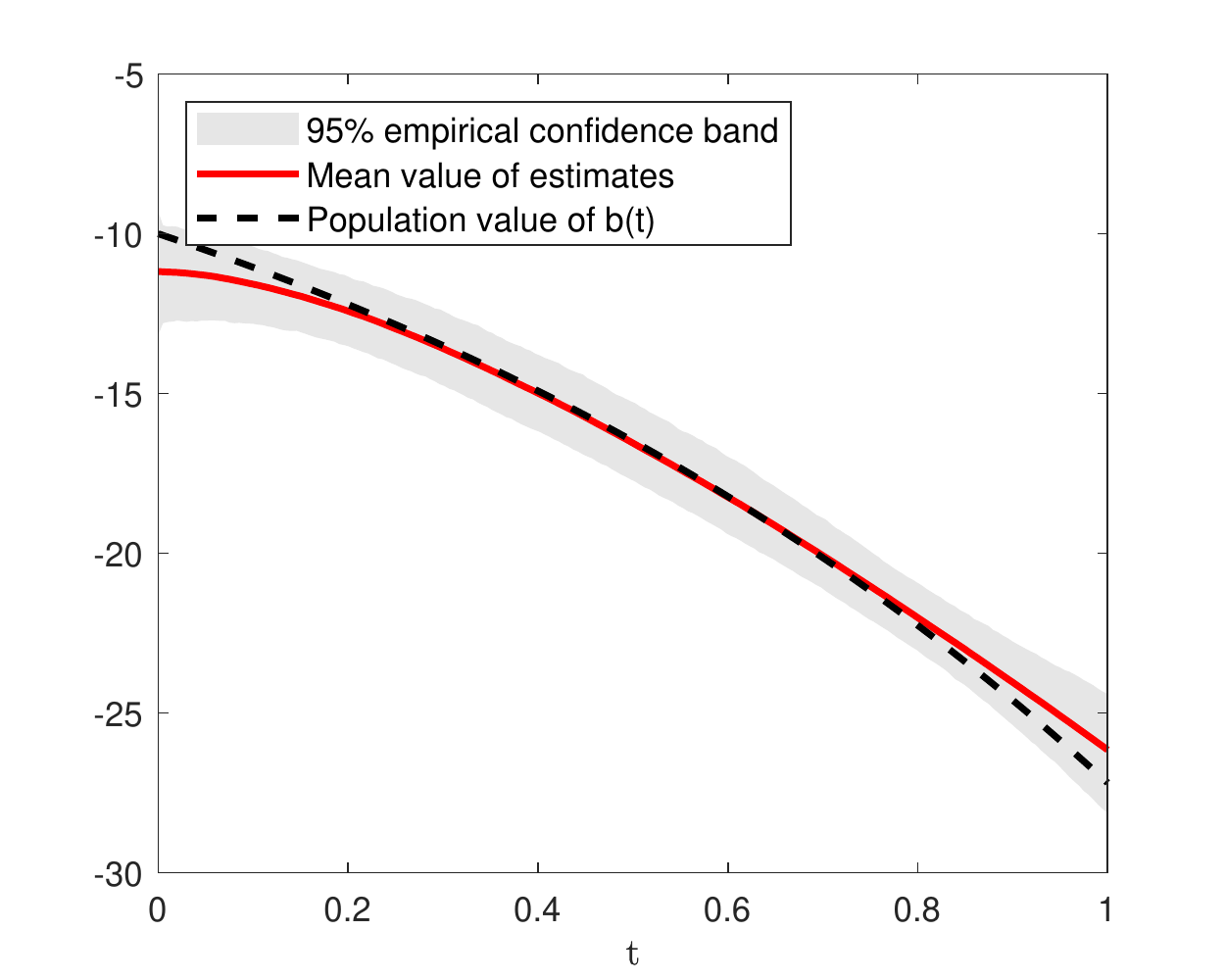}
		\caption{$\sigma=1,T=500$}
	\end{subfigure}
	\begin{subfigure}[b]{0.32\textwidth}
		\includegraphics[width=\textwidth]{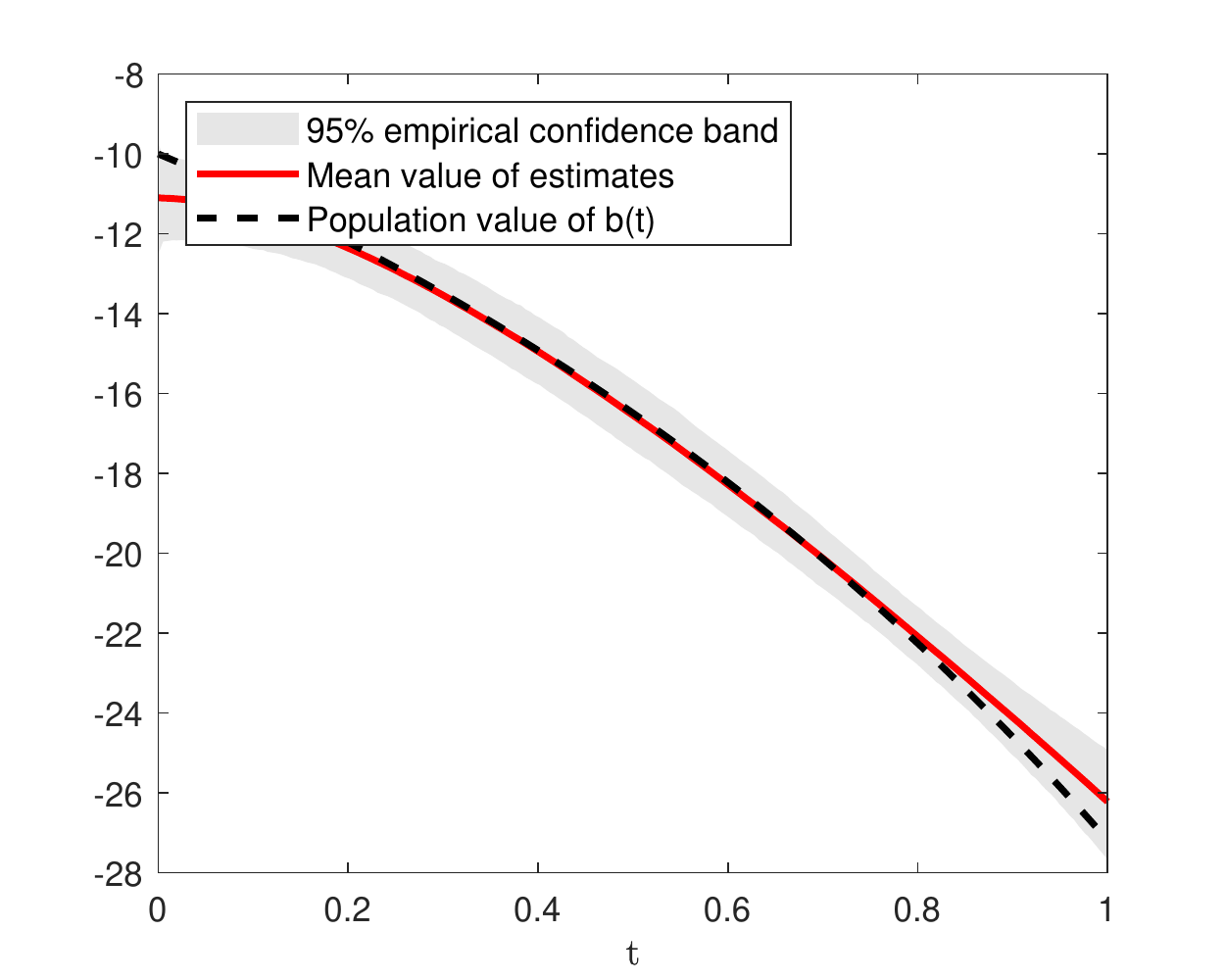}
		\caption{$\sigma=1,T=1000$}
	\end{subfigure}
	\caption{Summary of Monte Carlo experiments for $\beta(s)=-10\exp(s)$. Regularization parameter: $\alpha=10^{-6}$.}
	\label{fig:2_1}
\end{figure}

\begin{figure}
	\begin{subfigure}[b]{0.32\textwidth}
		\includegraphics[width=\textwidth]{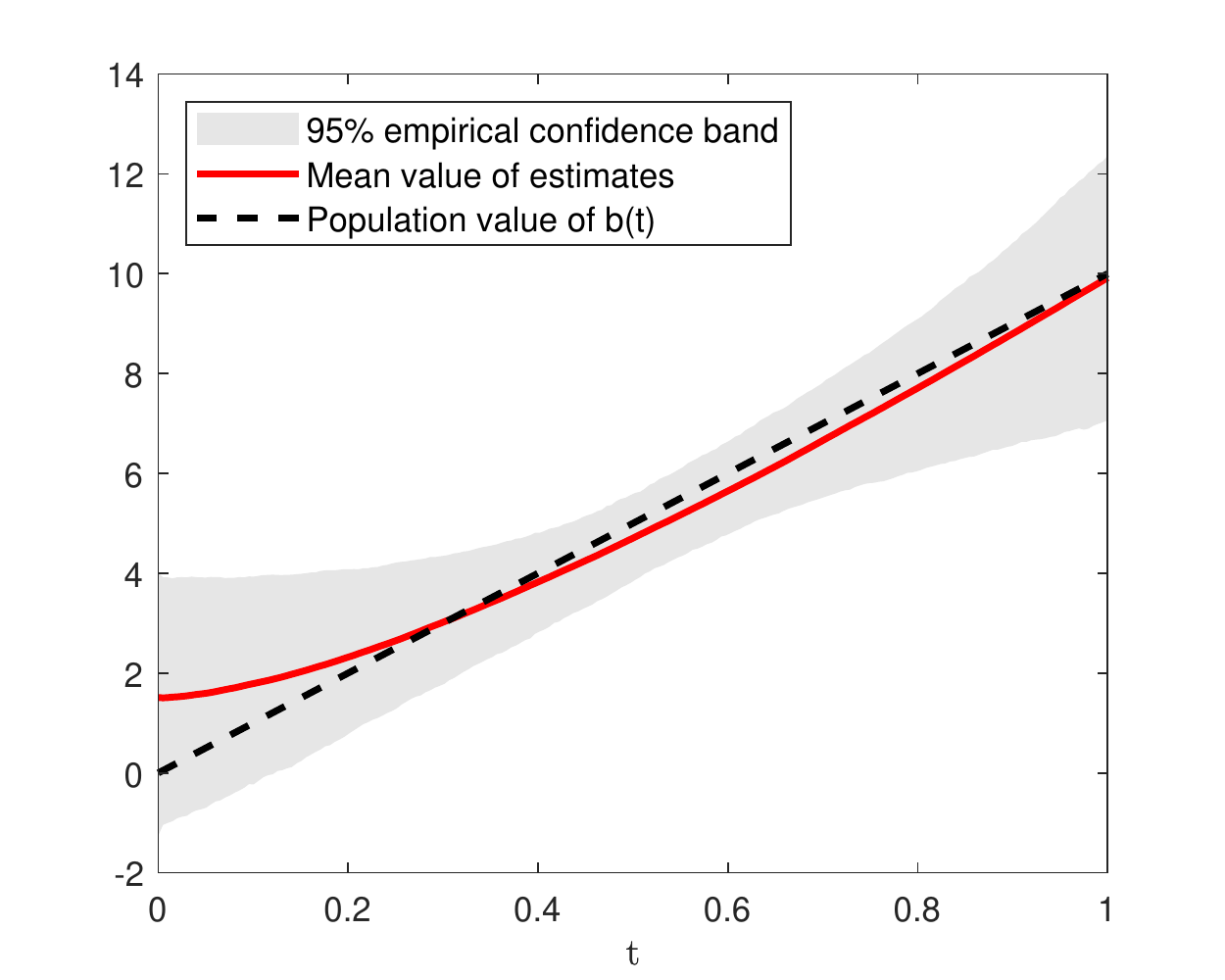}
		\caption{$\sigma=0.5,T=100$}
	\end{subfigure}
	\begin{subfigure}[b]{0.32\textwidth}
		\includegraphics[width=\textwidth]{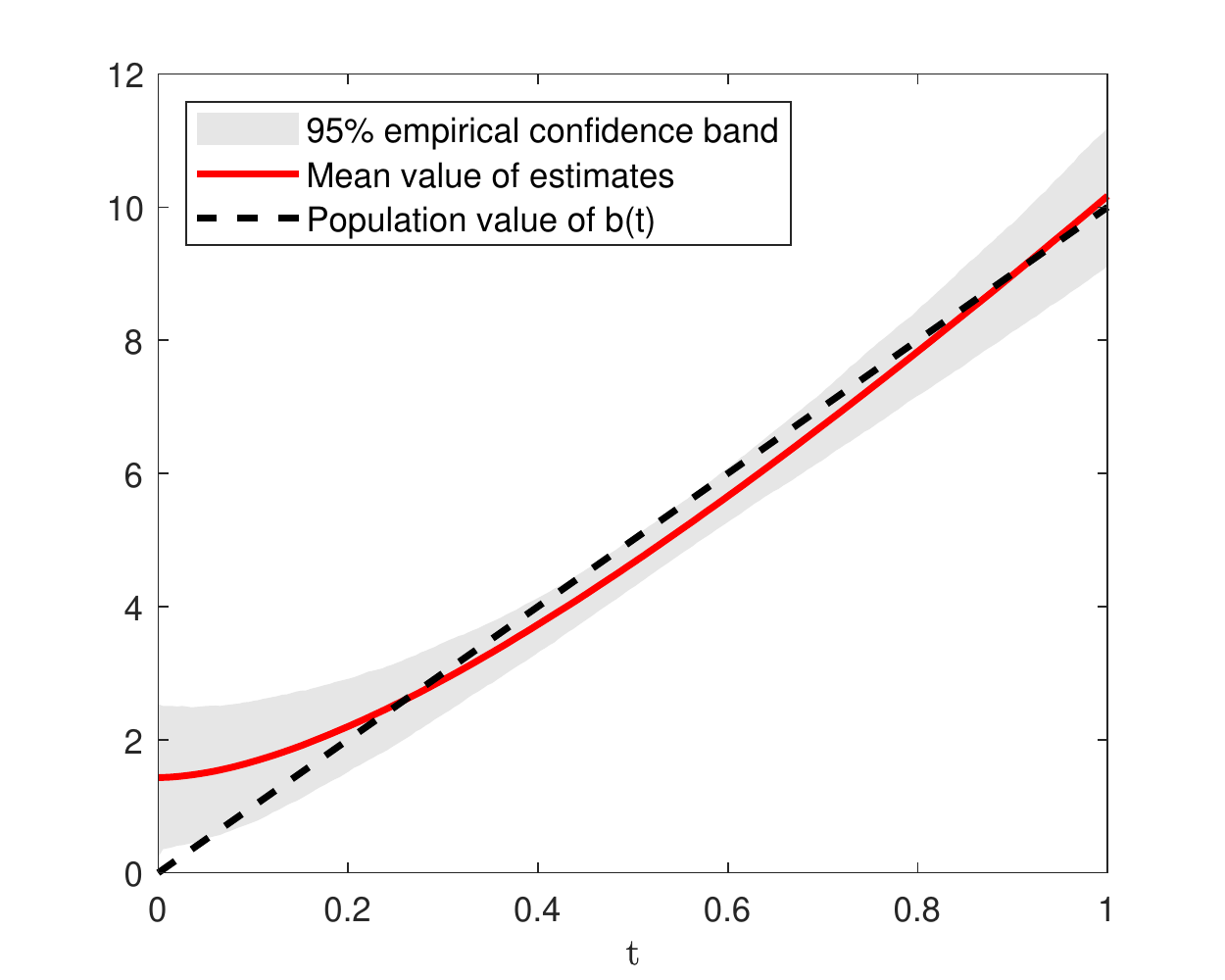}
		\caption{$\sigma=0.5,T=500$}
	\end{subfigure}
	\begin{subfigure}[b]{0.32\textwidth}
		\includegraphics[width=\textwidth]{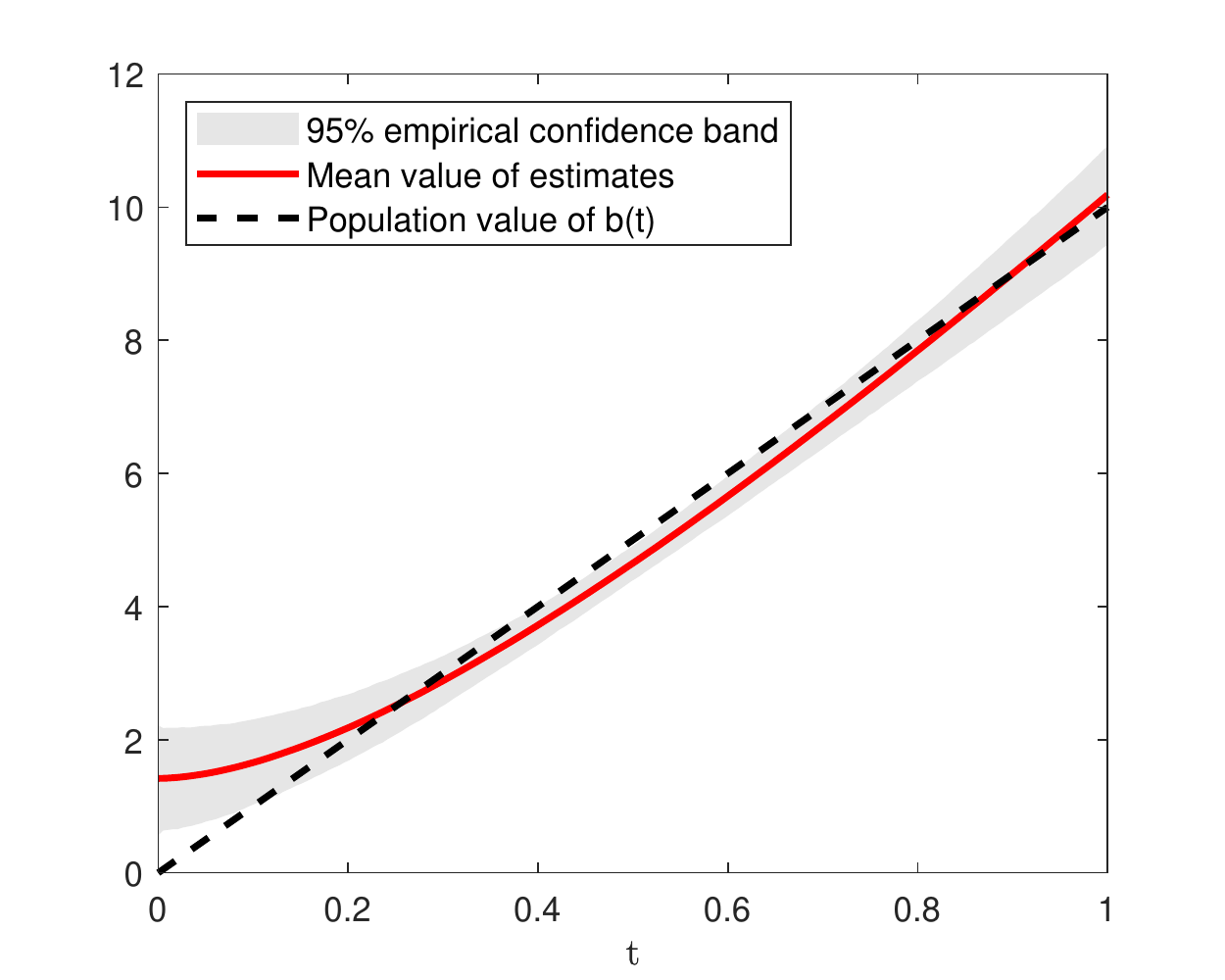}
		\caption{$\sigma=0.5,T=1000$}
	\end{subfigure}
	\begin{subfigure}[b]{0.32\textwidth}
		\includegraphics[width=\textwidth]{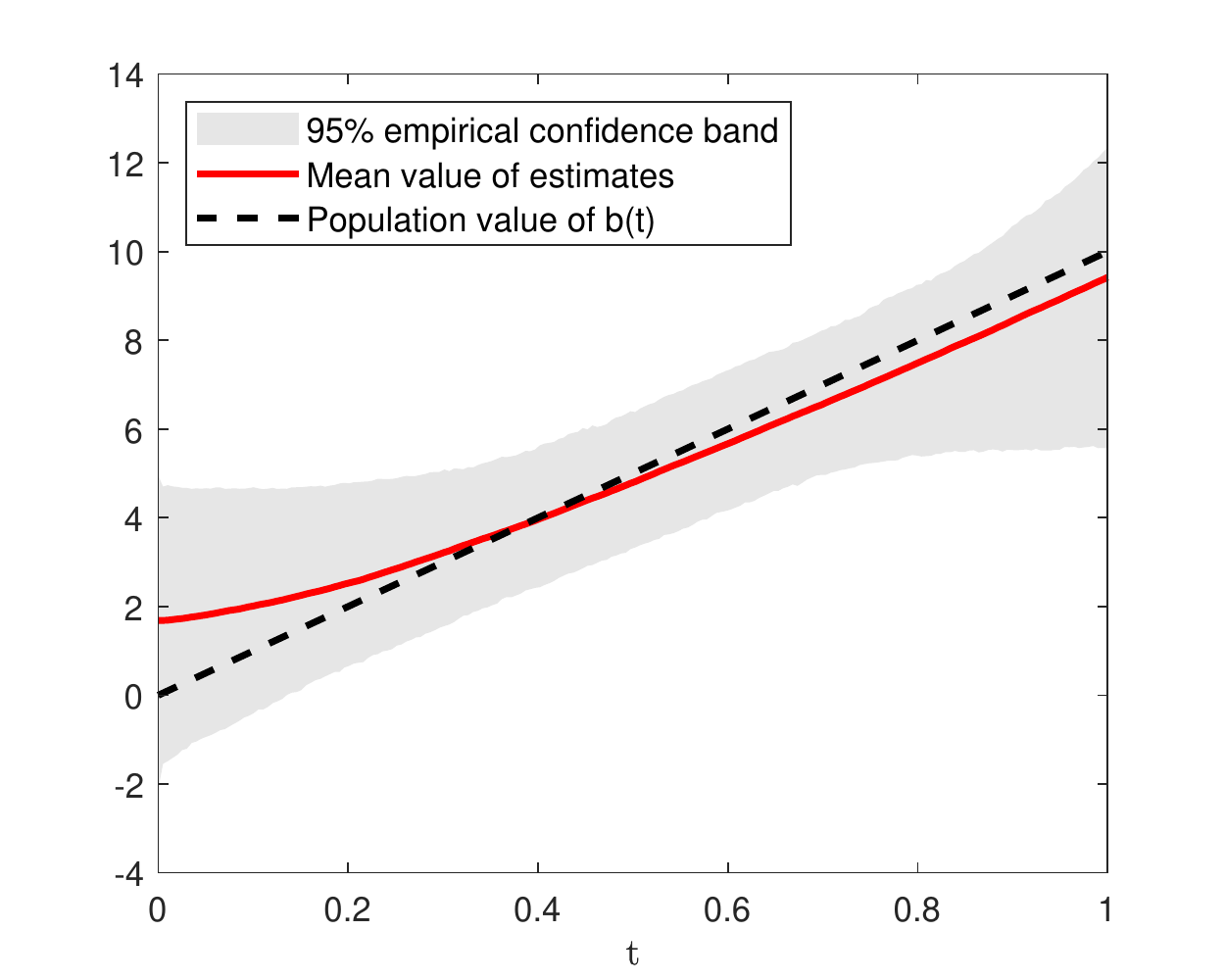}
		\caption{$\sigma=1,T=100$}
	\end{subfigure}
	\begin{subfigure}[b]{0.32\textwidth}
		\includegraphics[width=\textwidth]{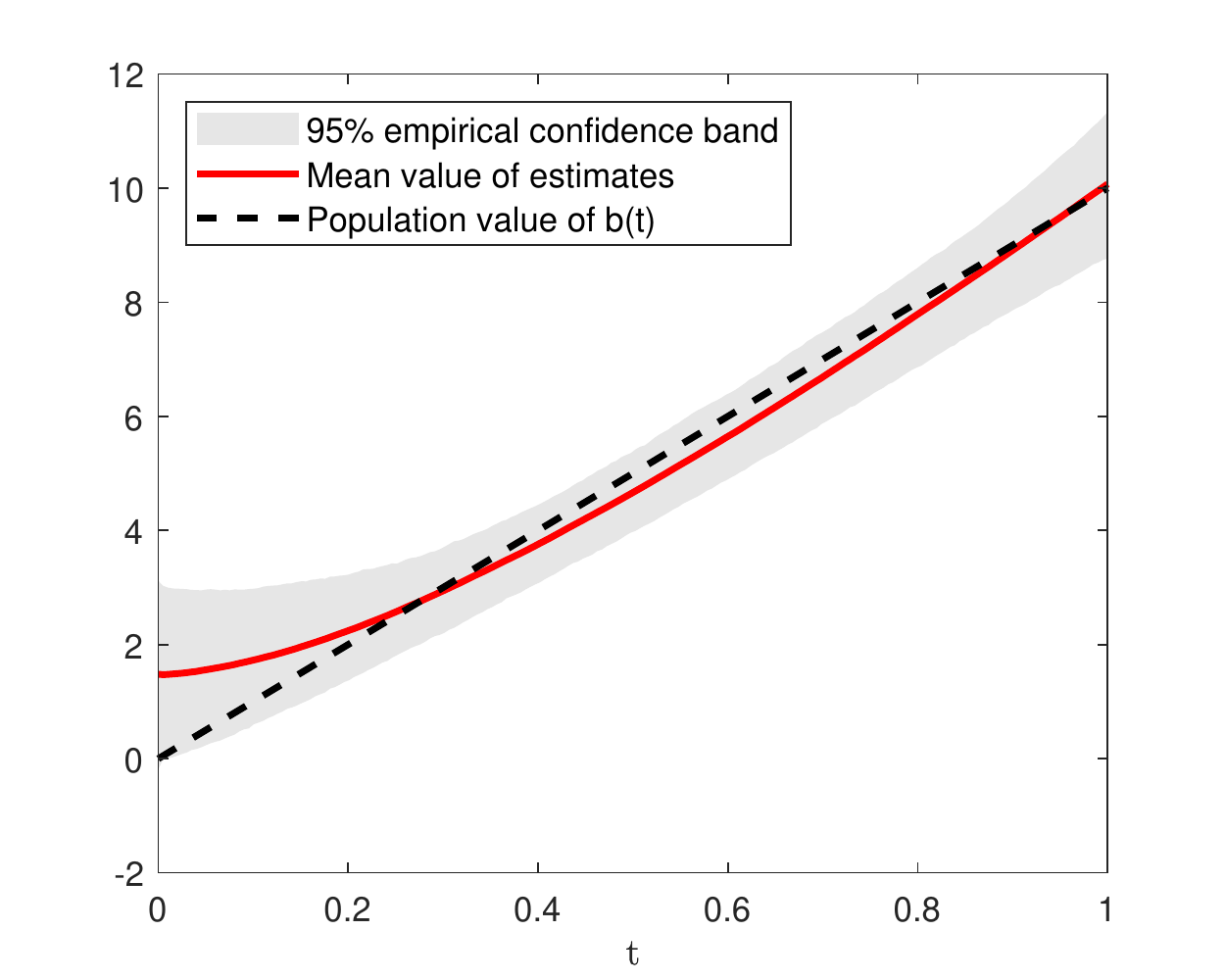}
		\caption{$\sigma=1,T=500$}
	\end{subfigure}
	\begin{subfigure}[b]{0.32\textwidth}
		\includegraphics[width=\textwidth]{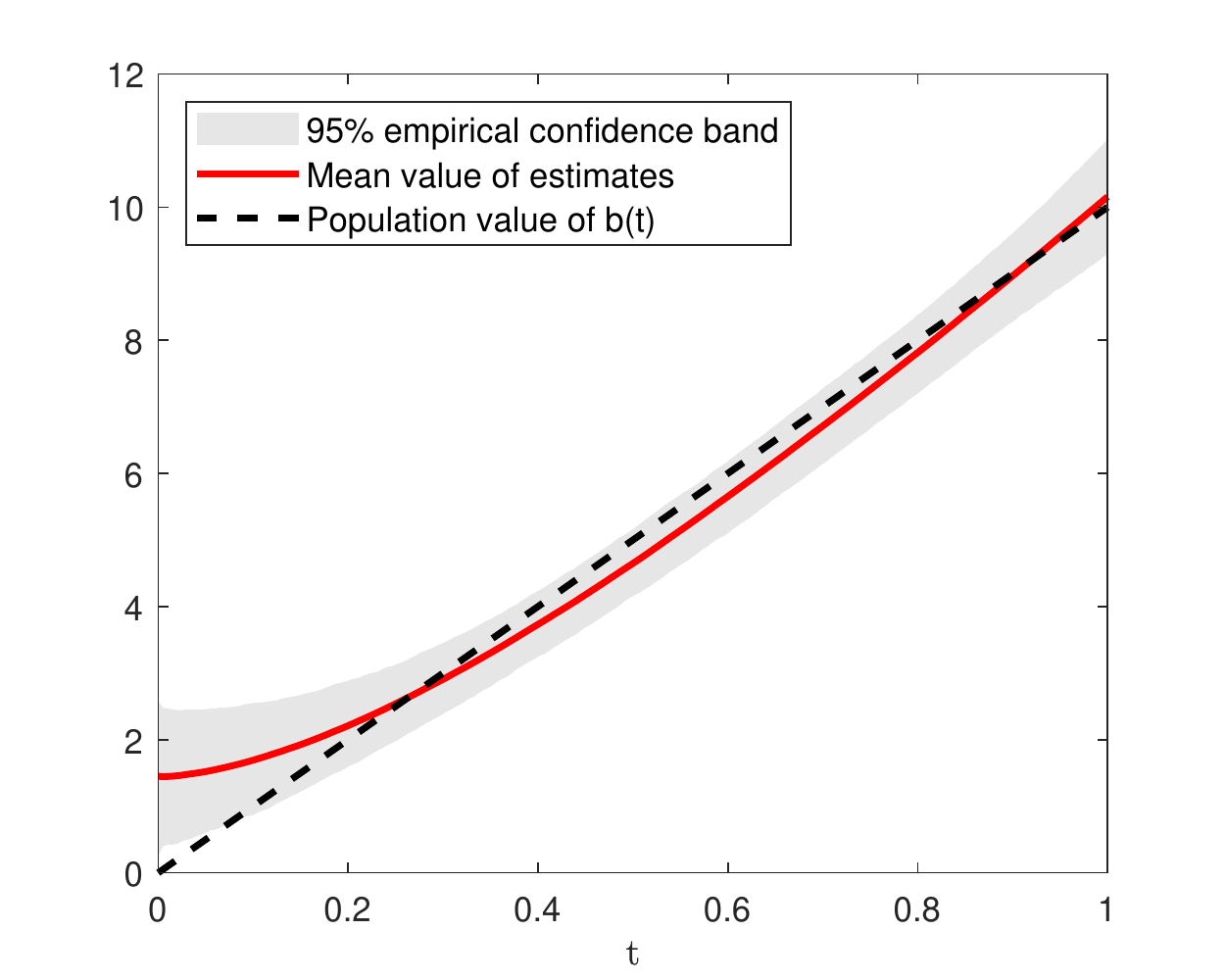}
		\caption{$\sigma=1,T=1000$}
	\end{subfigure}
	\caption{Summary of Monte Carlo experiments for $\beta(s)=10s$. Regularization parameter: $\alpha=10^{-6}$.}
	\label{fig:2_2}
\end{figure}

\section{Real time elasticity of electricity supply}\label{sec:application}
At the beginning of the 90s, electricity markets around the world were vertically integrated industries with prices set by regulators. Over the last 30 years, major countries experienced deregulation. Today, electricity is often sold at competitive spot markets where prices are determined according to the laws of supply and demand. Elasticities of supply and demand summarize the behavior of energy producers and consumers, inform market participants, and play an important role in the policy design, forecasting, and energy planning. The real-time elasticity of supply contains a piece of important information on seller's response to the intraday price fluctuations.\footnote{While there is an extensive literature on forecasting with intraday electricity data, see, e.g., \cite{aneiros2013functional} and references therein, the structural econometric analysis of the real-time electricity data received less attention, see \cite{benatia2017functional} and \cite{benatia2018functional} for notable exceptions. The latter paper studies the multi-unit electricity auction in New York and estimates the firm-level market power. It is also worth mentioning that the real-time price elasticities of demand have been previously estimated in \cite{patrick2001estimating} and \cite{lijesen2007real} relying on a different econometric methodology.}

Most of the electricity in Australia is generated, sold, and bought at the National Electricity Market (NEM), which is one of the largest interconnected electricity systems in the world. The NEM started operating as a wholesale spot market in December 1998. It supplies about 200 terawatt-hours of electricity to around 9 million customers each year reaching \$16.6 billion of trades in 2016-2017. The supply and the demand come from over 100 competitive generators and retailers participating in the market and are matched instantaneously in real time through a centrally coordinated dispatch process. Generators offer to supply a fixed amount of electricity at a specific time in the future and can resubmit subsequently the offered amount and price if needed. The Australian Energy Market Operator (AMEO) decides which generators will produce electricity to meet the demand in the most cost-efficient way.

We construct a new dataset using publicly available data from the AEMO and the Australian Bureau of Meteorology for the New South Wales in 1999-2018. The central pieces of the dataset are the daily aggregate quantities of the electricity sold at the spot market, intraday high-frequency prices measured each half an hour, and the average daily temperatures. The high-dimensional mixed-frequency IV regression model is
\begin{equation*}
	\log Q_{t} = \int_0^{24}\beta(s)\log P_t(s)\dx s + U_t,\qquad \E[U_t|W_t] = 0,
\end{equation*}
where $Q_t$ is the quantity sold on a day $t$, $P_t(s)$ is the price at time $s$ on a day $t$, and $W_t$ is an instrumental variable. To estimate the supply elasticity, we use the average daily temperature. Since the observed prices are measured with half an hour intervals, the regression equation is discretized as
\begin{equation*}
	\log Q_{t} = 0.5\sum_{j=1}^{48}\beta(s_j)\log P_t(s_j) + U_t,
\end{equation*}
where $s_j = 0.5j$ with $j=1,2,\dots,48$.

Figure~\ref{fig:q_and_p} displays the histogram of the natural logarithm of equilibrium quantities and the boxplot with equilibrium prices plotted against the hour. Figure~\ref{fig:temperature} displays the histogram of the temperature and the scatterplot with quantities plotted against the temperature. Marginal distributions seem to be well-behaved. The price series seems to be not stationary during the day with the median price peaking in the evening and plummeting during the night. There is also more volatility in the price in the evening. Note that such our assumptions do not rule out intraday nonstationarities. Figure~\ref{fig:temperature} (b) illustrates that the quantity sold is driven by heating and cooling demands. It is worth stressing that the temperature series is available at the daily frequency which is not allowed in \cite{florens2015instrumental}. At the same time, our mixed-frequency IV regression model allows to instrument intraday prices with the daily temperature series.
\begin{figure}
	\begin{subfigure}[b]{0.5\textwidth}
		\includegraphics[width=\textwidth]{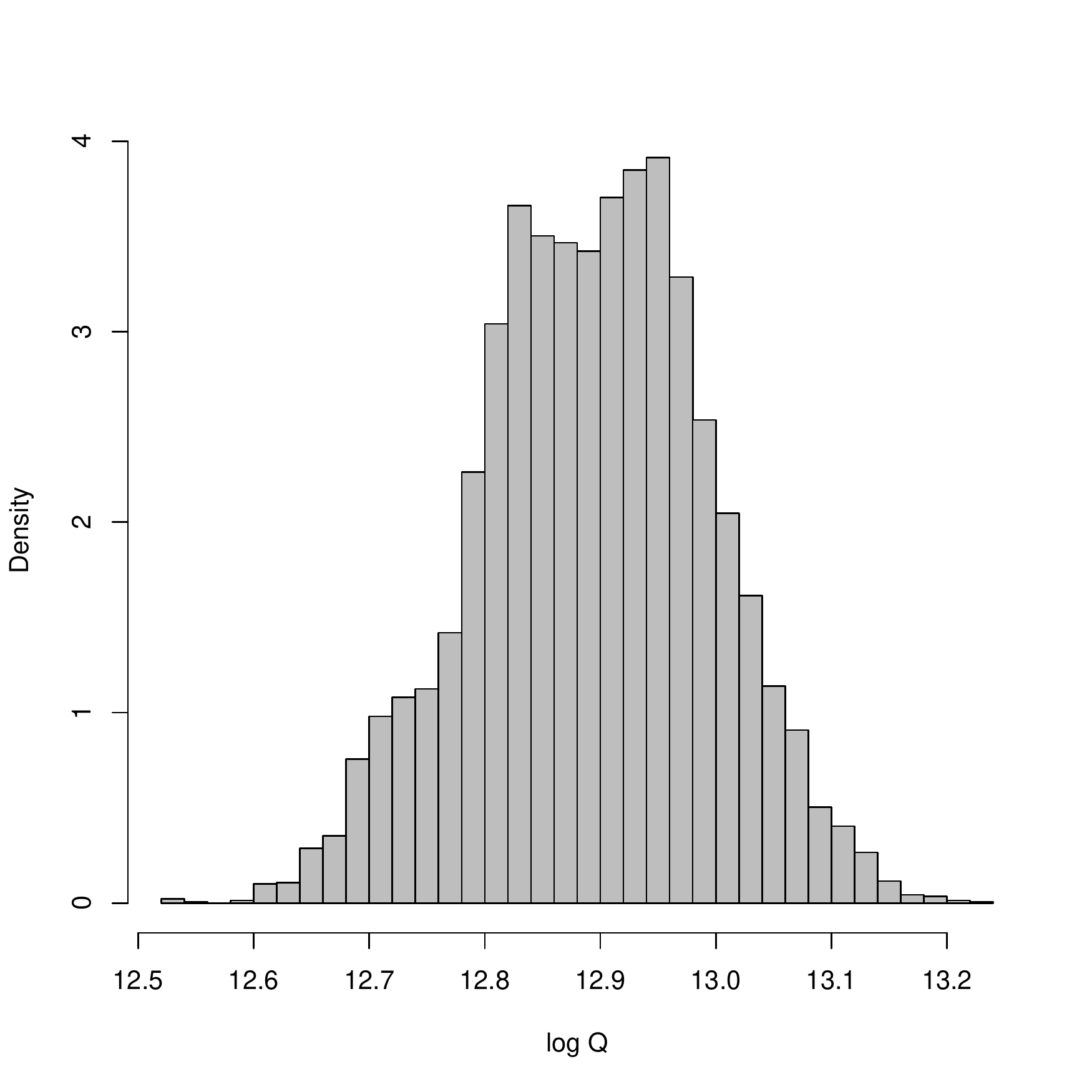}
		\caption{Distribution of quantities}
	\end{subfigure}
	\begin{subfigure}[b]{0.5\textwidth}
		\includegraphics[width=\textwidth]{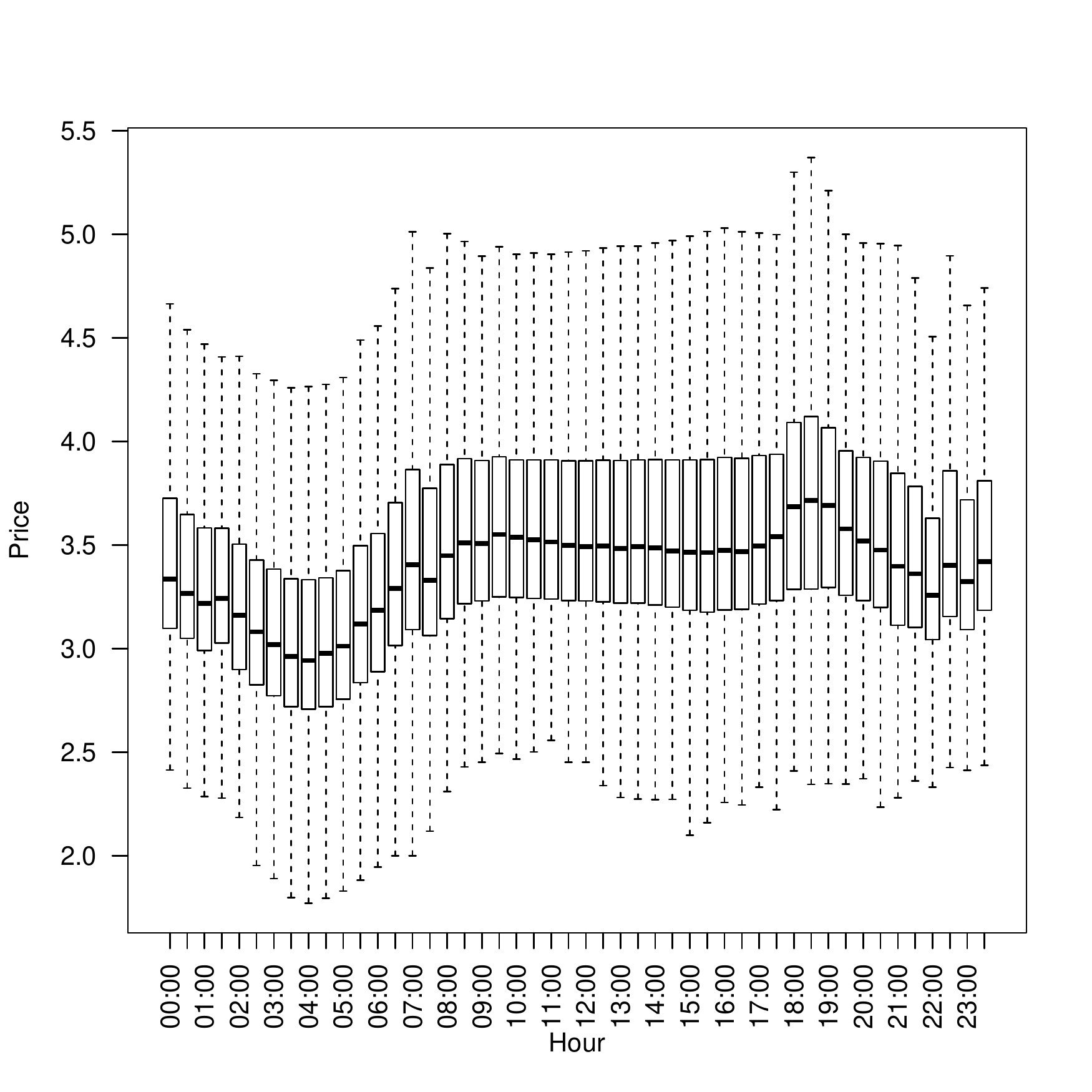}
		\caption{Distribution of prices}
	\end{subfigure}
	\caption{Quantities and Prices, in natural logarithms}
	\label{fig:q_and_p}
\end{figure}

\begin{figure}
	\begin{subfigure}[b]{0.5\textwidth}
		\includegraphics[width=\textwidth]{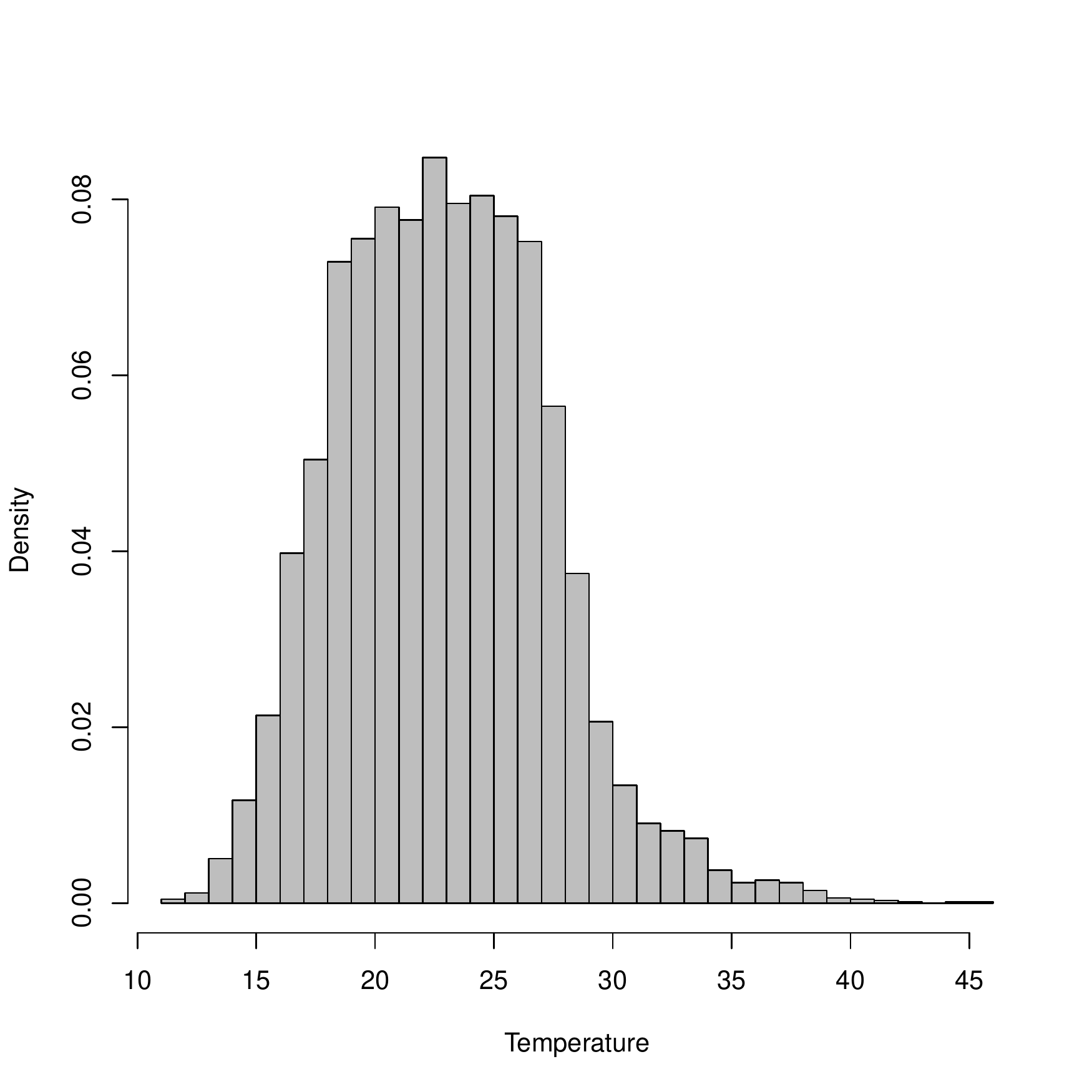}
		\caption{Temperature}
	\end{subfigure}
	\begin{subfigure}[b]{0.5\textwidth}
		\includegraphics[width=\textwidth]{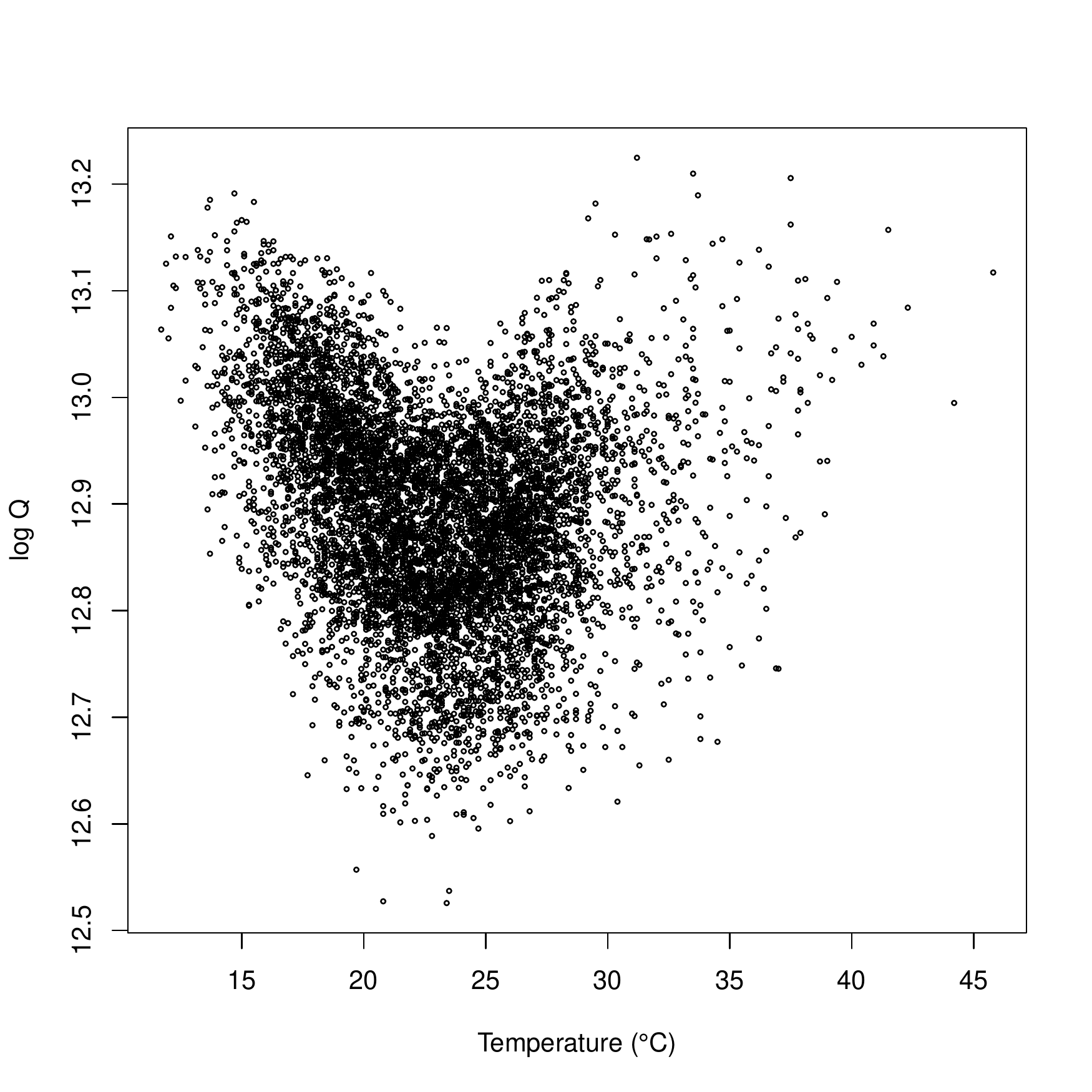}
		\caption{Temperature and Electricity}
	\end{subfigure}
	\caption{Temperature}
	\label{fig:temperature}
\end{figure}

To compute the estimator, we estimate the regularization parameter using a method similar to the one used in \cite{feve2010practice}. The method consists of minimizing the approximately scaled $L_2$ norm of the residual of the inverse problem 
\begin{equation*}
	RSS(\alpha) = \alpha^{-1}\|\hat K\hat\beta_\alpha - \hat g\|^2,
\end{equation*}
where $\hat\beta_\alpha = (\alpha I + \hat K^*\hat K)^{-1}\hat K^*\hat g$. In our case, the minimum is reached at $\alpha^*=2.16\times 10^{-3}$ as be seen from Figure~\ref{fig:estimates}, panel (a).

\begin{figure}
	\begin{subfigure}[b]{0.5\textwidth}
		\includegraphics[width=\textwidth]{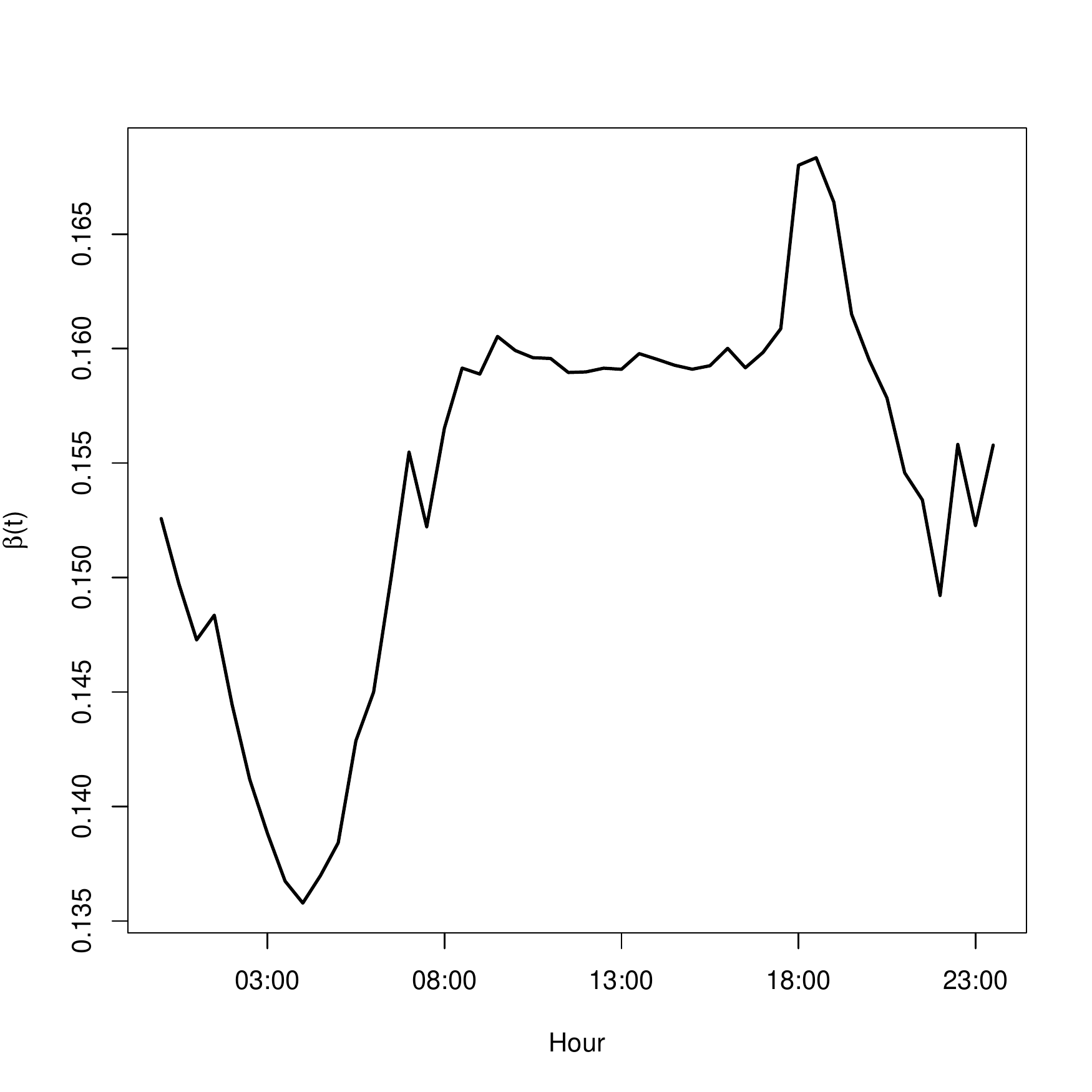}
		\caption{Elasticity of electricity supply}
	\end{subfigure}
	\begin{subfigure}[b]{0.5\textwidth}
		\includegraphics[width=\textwidth]{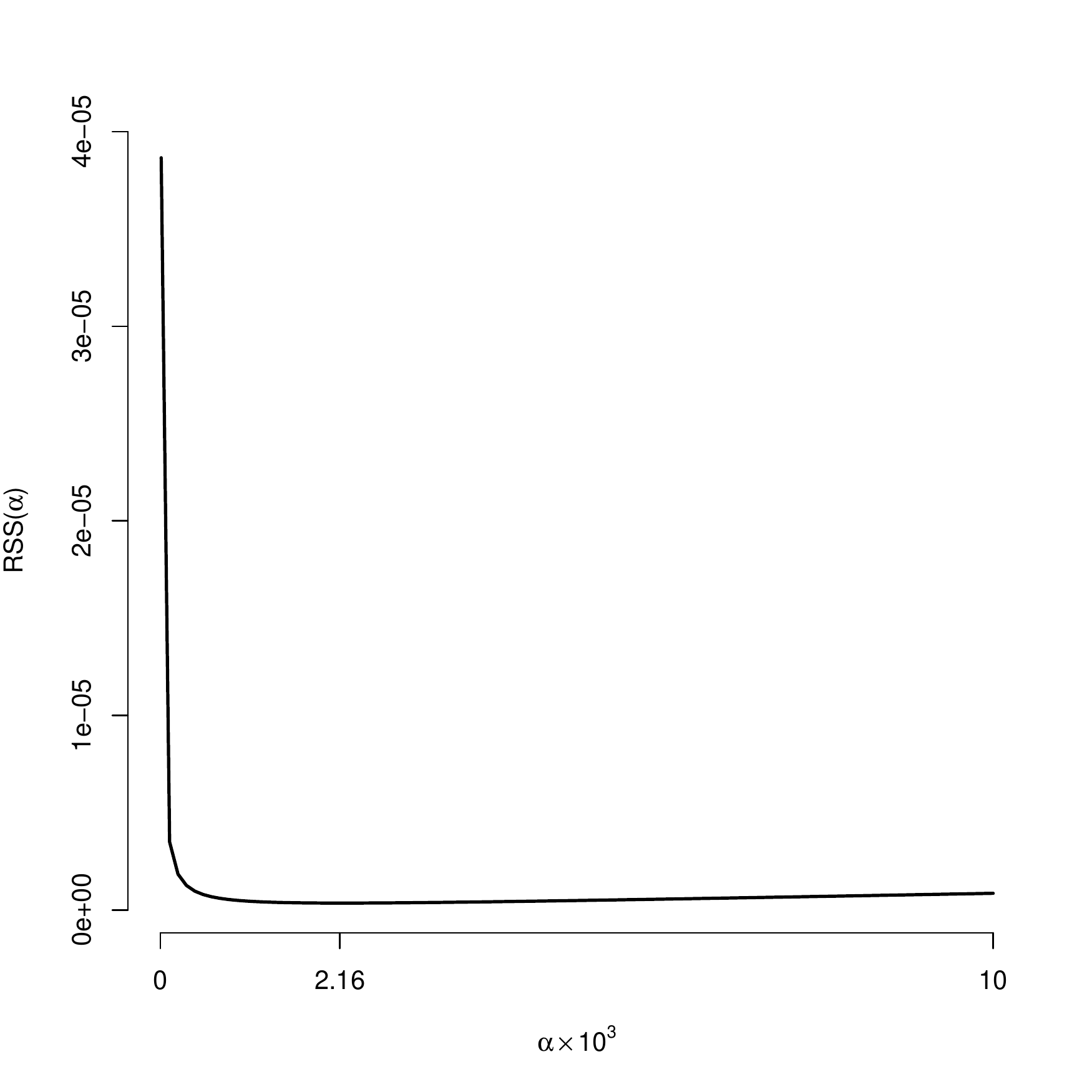}
		\caption{Residual curve}
	\end{subfigure}
	\caption{Estimator and optimal regularization parameter}
	\label{fig:estimates}
\end{figure}

Figure~\ref{fig:estimates}, panel (b) displays the estimated intraday elasticity of supply using our high-dimensional mixed-frequency IV regression. We find that depending on the hour, the price elasticity of supply ranges between 0.135 and 0.165. The elasticity is the highest in the evening, around 6 pm and the lowest during the night. The supply appears to have the real-time price elasticity of a similar order of magnitude as the demand.\footnote{\cite{patrick2001estimating} find real-time demand elasticities between 0 and -0.27 for 5 industrial sectors in UK.} The relatively inelastic supply may probably be attributed to the fact that the market participants are allowed to hedge financial risks and the difficulty to adjust the electricity production in real time.

\section{Conclusions}\label{section:conclusions}
This paper introduces a novel high-dimensional mixed-frequency IV regression and contributes to the growing literature on high-dimensional and mixed-frequency data. We show that the slope parameter of the high-dimensional endogenous regressor can be identified and accurately estimated leveraging on an instrumental variable observed at a low-frequency only.

We characterize the identifying condition in the model and study the statistical properties of the Tikhonov-regularized estimator with time series data. The mixed-frequency IV estimator has a closed-form expression and is easy and fast to compute numerically. Our statistical analysis does not restrict the number of high-frequency observations of the process and can handle the number of covariates increasing with the sample size even faster than exponentially.

In our empirical application, we estimate the real-time price elasticity of supply at the Australian electricity spot market. We find that the supply is relatively inelastic and that its elasticity is heterogeneous throughout the day. To conclude, we note that our identification strategy with a low-frequency IV can also be applied to the instrumental variable model of \cite{benatia2017functional} with a high-frequency dependent variable.

\newpage
\setcounter{page}{1}
\setcounter{section}{0}
\setcounter{equation}{0}
\setcounter{table}{0}
\setcounter{figure}{0}
\renewcommand{\theequation}{A.\arabic{equation}}
\renewcommand\thetable{A.\arabic{table}}
\renewcommand\thefigure{A.\arabic{figure}}
\renewcommand\thesection{A.\arabic{section}}
\renewcommand\thepage{Appendix - \arabic{page}}
\renewcommand\thetheorem{A.\arabic{theorem}}

\begin{center}
	{\LARGE\textbf{APPENDIX}}	
\end{center}
\bigskip
\section{Proofs}\label{section:appendix}
\paragraph{Notation:} We use $L_2(S)$ to denote the space of functions on $S\subset \R^d$, square-integrable with respect to the Lebesgue measure. We endow the space $L_2(S)$ with the natural inner product $\langle \beta,\gamma\rangle=\int_{ S} \beta(s)\gamma(s)\dx s$ and the norm $\|\beta\|=\sqrt{\langle\beta,\beta\rangle}$ for all $\beta,\gamma\in L_2(S)$. Any vector $a\in\R^m$ should be considered as a column-vector and can be written as $a=(a_j)_{1\leq j\leq m}$. For a bounded linear operator $K:\mathcal{E}\to\mathcal{F}$ between the two Hilbert spaces $\mathcal{E}$ and $\mathcal{F}$, let $\|K\|_{\infty} = \sup_{\|\phi\|\leq 1}\|K\phi\|$ denote its operator norm. Let $\sigma(K^*K)$ denote the spectrum of the corresponding self-adjoint operator $K^*K$. The $m\times T$ matrix $A$ is written by enumerating all its elements $A=(A_{j,t})_{\substack{1\leq j\leq m\\1\leq t\leq T}}$. If $m=T$, then we simply write $A = (A_{ij})_{1\leq i,j\leq m,}$. We use
\begin{equation*}
	C_L^\kappa[0,1] = \left\{f:[0,1]\to \R:\; \max_{1\leq k\leq \lfloor\kappa\rfloor}\|f^{(k)}\|_\infty\leq L,\quad\sup_{s\ne s'}\frac{|f^{(\lfloor\kappa\rfloor)}(s) - f^{(\lfloor\kappa\rfloor)}(s')|}{|s-s'|^{\kappa - \lfloor\kappa\rfloor}} \leq L \right\}
\end{equation*}
to denote the space of H\"{o}lder continuous functions with common parameters $\kappa,L>0$. Lastly, for $a,b\in\R$, put $a\vee b = \max\{a,b\}$.

To prove Theorem~\ref{thm:main_tikhonov}, we need two auxiliary lemmas. The first lemma bounds the expected norm of the sample mean of a covariance stationary zero-mean $L_2(S)$-valued stochastic process $(X_t)_{t\in\Z}$ by the norm of its auto-covariance function $\gamma_h$.
\begin{lemma}\label{lemma:parametric_rate_ts}
	Suppose that $(X_t)_{t\in\Z}$ is a zero-mean covariance stationary process in $L_2(S)$ with absolutely summable autocovariance function
	\begin{equation*}
	\sum_{h\in\mathbf{Z}} \|\gamma_h\|_1 < \infty,
	\end{equation*}
	where $\|\gamma_h\|_1=\int_S|\gamma_h(s,s)|\dx s$. Then
	\begin{equation*}
		\E\left\|\frac{1}{T}\sum_{t=1}^TX_t\right\|^2 \leq \frac{1}{T}\sum_{h\in\mathbf{Z}} \|\gamma_h\|_1.
	\end{equation*}
\end{lemma}
\begin{proof}
	We have
	\begin{equation*}
	\begin{aligned}
		\E\left\|\frac{1}{T}\sum_{t=1}^TX_t\right\|^2 & = \frac{1}{T^2}\E\left\langle\sum_{t=1}^TX_t,\sum_{k=1}^TX_k\right\rangle \\
		& = \frac{1}{T^2}\sum_{t,k=1}^T\int_S\E[X_t(s)X_k(s)]\dx s \\
		& = \frac{1}{T^2}\sum_{t,k=1}^T\int_S\gamma_{t-k}(s)\dx s \\
		& = \frac{1}{T}\sum_{|h|<T}\frac{T-|h|}{T}\int_S\gamma_h(s,s)\dx s \\
		& \leq \frac{1}{T}\sum_{h\in\mathbf{Z}}\int_S\left|\gamma_h(s,s)\right|\dx s,
	\end{aligned}
	\end{equation*}
	where the second line follows by the bilinearity of the inner product and Fubini's theorem and the third under the covariance stationarity.
\end{proof}

The following lemma allows controlling estimation errors appearing in the proof of Theorem~\ref{thm:main_tikhonov} in terms of more primitive quantities.
\begin{lemma}\label{lemma:contiv}
	Suppose that $\hat k,k,\hat r,r,\beta$ are square-integrable. Then
	\begin{equation*}
		\E\left\|\hat K - K\right\|^2_\infty \leq \E\left\|\hat k - k\right\|^2
	\end{equation*}
	and
	\begin{equation*}
	\E\left\|\hat r - \hat K\beta\right\|^2 \leq 2\E\left\|\hat r - r\right\|^2 + 2\|\beta\|^2\E\left\|\hat k - k\right\|^2.
	\end{equation*}
\end{lemma}
\begin{proof}
	By the definition of the operator norm and the Cauchy-Schwartz inequality
	\begin{equation}\label{eq:hsboudn}
	\begin{aligned}
	\E\left\|\hat K - K\right\|^2_\infty & = \E\left[\sup_{\|\phi\|\leq1}\left\|\hat K\phi - K\phi\right\|^2\right] \\
	& = \E\left[\sup_{\|\phi\|\leq1}\int\left|\int\phi(s)\left(\hat k(s,u) - k(s,u)\right)\dx s\right|^2\dx u\right] \\
	& \leq \E\left\|\hat k - k\right\|^2.
	\end{aligned}
	\end{equation}
	
	For the second part, use $r=K\beta$, $\|a+b\|^2\leq 2\|a\|^2+2\|b\|^2$, $\|K\beta\| \leq\|K\|_\infty\|\beta\|$, and the estimate in Eq.~\ref{eq:hsboudn}
	\begin{equation*}
	\begin{aligned}
	\E\left\|\hat r - \hat K\beta\right\|^2 & \leq 2\E\left\|\hat r - r\right\|^2 + 2\E\|\hat K\beta - K\beta\|^2 \\
	& \leq 2\E\left\|\hat r - r\right\|^2 + 2\|\beta\|^2\E\left\|\hat k - k\right\|^2.
	\end{aligned}
	\end{equation*}
\end{proof}

\begin{proof}[Proof of Theorem~\ref{thm:main_tikhonov}]
	The proof is based on the following decomposition
	\begin{equation*}
	\hat \beta - \beta = R_1 + R_2 + R_3 + R_4
	\end{equation*}
	with
	\begin{equation*}
	\begin{aligned}
	R_1 & = (\alpha I + \hat K^*\hat K)^{-1}\hat K^*(\hat r - \hat K\beta) \\
	R_2 & = \alpha(\alpha I + \hat K^*\hat K)^{-1}\hat K^*(\hat K - K)(\alpha I + K^* K)^{-1}\beta \\
	R_3 & = \alpha(\alpha I + \hat K^*\hat K)^{-1}(\hat K^* - K^*)K(\alpha I + K^* K)^{-1}\beta \\
	R_4 & = (\alpha I + K^*K)^{-1}K^*K\beta- \beta.
	\end{aligned}
	\end{equation*}
	To see that this decomposition holds, note that
	\begin{equation*}
	\begin{aligned}
	R_2 + R_3 & = \alpha(\alpha I + \hat K^*\hat K)^{-1}\left[\hat K^*\hat K - K^*K\right](\alpha I + K^* K)^{-1}\beta \\
	& = \alpha(\alpha I + \hat K^*\hat K)^{-1}\left[(\alpha I + \hat K^*\hat K) - (\alpha I + K^*K)\right](\alpha I + K^* K)^{-1}\beta \\ 
	& = \alpha(\alpha I + K^* K)^{-1}\beta - \alpha(\alpha I + \hat K^*\hat K)^{-1} \beta \\
	& = \left[I - \alpha(\alpha I + \hat K^*\hat K)^{-1}\right] \beta + \left[\alpha(\alpha I + K^* K)^{-1} - I\right]\beta\\
	& = (\alpha I + \hat K^*\hat K)^{-1}\hat K^*\hat K\beta - (\alpha I + K^*K)^{-1}K^*K\beta.
	\end{aligned}
	\end{equation*}	
	Therefore,
	\begin{equation*}
	\E\left\|\hat \beta - \beta\right\|^2 \leq 4\E\|R_1\|^2 + 4\E\|R_2\|^2 + 4\E\|R_3\|^2 + 4\E\|R_4\|^2.
	\end{equation*}
	
	The fourth term is a regularization bias and its order follows directly from the Assumption~\ref{as:source} and the isometry of the functional calculus
	\begin{equation*}
	\begin{aligned}
	\E\|R_4\|^2 & = \left\|\left[(\alpha I + K^*K)^{-1}K^*K - I\right]\beta\right\|^2 \\
	& \leq \left\|\left[I - (\alpha I + K^*K)^{-1}K^*K\right](K^*K)^{\gamma}\right\|^2_\infty R \\
	& = \sup_{\lambda\in\sigma(K^*K)}\left|\left(1 - \frac{\lambda}{\alpha + \lambda}\right)\lambda^{\gamma}\right|^2R \\
	& = \sup_{\lambda\in\sigma(K^*K)}\left|\frac{\lambda^{\gamma}}{\alpha + \lambda}\right|^2\alpha^2R.
	\end{aligned}
	\end{equation*}
	We can have two cases depending on the value of $\gamma>0$. For $\gamma\in(0,1)$, the function $\lambda\mapsto\frac{\lambda^{\gamma}}{\alpha+\lambda}$ admits maximum at $\lambda = \frac{\gamma}{1-\gamma}\alpha$. For $\gamma\geq1$, the function $\lambda\mapsto \frac{\lambda^\gamma}{\alpha + \lambda}$ is strictly increasing on $[0,\infty)$, attaining maximum at the end of the spectrum $\lambda=\|K^*K\|_\infty$.  Therefore, since $\gamma^{\gamma}(1-\gamma)^{1-\gamma}\leq 1,\gamma\in(0,1)$, we have
	\begin{equation*}
	\sup_{\lambda\in\sigma(K^*K)}\frac{\lambda^{\gamma}}{\alpha + \lambda} \leq \begin{cases}
	\|K^*K\|^{\gamma-1}, & \gamma\geq1	 \\
	\alpha^{\gamma-1}, & \gamma\in(0,1).
	\end{cases}
	\end{equation*}
	This gives $\E\|R_4\|^2 \leq \alpha^{2\gamma}R$ since $\gamma\in(0,1]$.
	
	Similar computations\footnote{Note that $\hat K$ is a finite-rank operator, hence, compact.} give
	\begin{equation*}
	\begin{aligned}
		\E\|R_1\|^2 & \leq \E\left\|(\alpha I + \hat K^*\hat K)^{-1}\hat K^*\right\|^2_\infty\left\|\hat r - \hat K\beta\right\|^2 \\
		& = \E\left[\sup_{\lambda\in\sigma(\hat K^*\hat K)}\left|\frac{\lambda^{1/2}}{\alpha+\lambda}\right|^2\left\|\hat r - \hat K\beta\right\|^2\right] \\
		& \leq \frac{2}{\alpha}\left(\E\left\|\hat r - r\right\|^2 + \|\beta\|^2\E\left\|\hat k - k\right\|^2\right),
	\end{aligned}
	\end{equation*}
	where the last inequality follows by Lemma~\ref{lemma:contiv}.
	
	Next,
	\begin{equation*}
	\begin{aligned}
	\E\|R_2\|^2 & \leq \E\left\|(\alpha I + \hat K^*\hat K)^{-1}\hat K^*\right\|_\infty^2\left\|\hat K  - K\right\|^2_\infty\left\|\alpha(\alpha I + K^*K)^{-1}(K^*K)^{\gamma}\right\|_\infty^2 R\\
	& \leq \E\left[\sup_{\lambda\in\sigma(\hat K^*\hat K)}\left|\frac{\lambda^{1/2}}{\alpha + \lambda}\right|^2\left\|\hat K - K\right\|^2_\infty \sup_{\lambda\in\sigma(K^*K)}\left|\frac{\lambda^{\gamma}}{\alpha + \lambda}\right|^2\alpha^2R\right]\\
	& \leq \frac{\alpha^{2\gamma}}{\alpha}\E\left\|\hat k - k\right\|^2R,
	\end{aligned}
	\end{equation*}
	where the last inequality follows by Lemma~\ref{lemma:contiv}.
	
	Likewise, for the third term, we have
	\begin{equation*}
	\begin{aligned}
	\E\|R_3\|^2 & = \E\left\|(\alpha I + \hat K^*\hat K)^{-1}(\hat K^* - K^*)\alpha K(\alpha I + K^*K)^{-1}\beta\right\|^2 \\
	& \leq \E\left\|(\alpha I + \hat K^*\hat K)^{-1}\right\|^2_\infty\left\|\hat K^*  - K^*\right\|^2_\infty \left\|K(\alpha I + K^*K)^{-1}(K^*K)^{\gamma}\right\|_\infty^2\alpha^2R \\
	& = \E\left[\sup_{\lambda\in\sigma(\hat K^*\hat K)}\left|\frac{1}{\alpha + \lambda}\right|^2\left\|\hat K^* - K^*\right\|^2_\infty\sup_{\lambda\in\sigma(K^*K)}\left|\frac{\lambda^{\gamma+1/2}}{\alpha + \lambda}\right|^2\alpha^2R\right] \\
	& \leq \sup_{\lambda\in\sigma(K^*K)}\left|\frac{\lambda^{\gamma+1/2}}{\alpha + \lambda}\right|^2\E\left\|\hat K^* - K^*\right\|^2_\infty R \\
	& \leq \frac{\alpha^{2\gamma\wedge 1}}{\alpha}\|K^*K\|^{(2\gamma - 1)\vee0}\E\left\|\hat k - k\right\|^2R.
	\end{aligned}
	\end{equation*}
	
	Lastly, let $\gamma_h^{(1)}$ be the autocovariance function of the process $(Y_t\Psi(.,W_t))_{t\in\Z}$ and let $\gamma_h^{(2)}$ be the autocovariance function of the process $(Z_t(.),\Psi(.,W_t))_{t\in\Z}$. Put
	\begin{equation*}
	\eta_1 = \sum_{h\in\Z}\|\gamma_{h}^{(1)}\|_1 \qquad \text{and}\qquad \eta_2 = \sum_{h\in\Z}\|\gamma_{h}^{(2)}\|_1.
	\end{equation*}
	Then under Assumption~\ref{as:data_ts} by Lemma~\ref{lemma:parametric_rate_ts}
	\begin{equation*}\label{eq:ghat}
	\begin{aligned}
	\E\left\|\hat r - r\right\|^2 & = \E\left\|\frac{1}{T}\sum_{t=1}^T\left\{Y_t\Psi(.,W_t) - \E[Y_t\Psi(.,W_t)]\right\}\right\|^2 \\
	& \leq \frac{\eta_1}{T}
	\end{aligned}
	\end{equation*}
	and
	\begin{equation*}\label{eq:khat}
	\begin{aligned}
	\E\left\|\hat k - k\right\|^2 & = \E\left\|\frac{1}{T}\sum_{t=1}^T\left\{Z_t(.)\Psi(.,W_t) - \E[Z(.)\Psi(.,W)]\right\}\right\|^2\\
	& \leq \frac{\eta_2}{T}.
	\end{aligned}
	\end{equation*}
	Combining all the estimates, we obtain
	\begin{equation}\label{eq:constant}
	\E\left\|\hat \beta - \beta\right\|^2 \leq \frac{8(\eta_1 + \eta_2\|\beta\|^2)}{\alpha T} + \frac{4R\eta_2}{\alpha T}\left(\alpha^{2\gamma} + \alpha^{2\gamma\wedge 1}\|K^*K\|^{(2\gamma-1)\vee 0}\right) + 4\alpha^{2\gamma}R.
	\end{equation}
\end{proof}

\begin{proof}[Proof of Theorem~\ref{thm:main_tikhonov_discretization}]
	Decompose
	\begin{equation*}
	\hat\beta_m - \beta = \hat\beta_m - \hat\beta + \hat\beta - \beta.
	\end{equation*}
	By Theorem~\ref{thm:main_tikhonov}, we know that $\E\|\hat\beta - \beta\|^2 = O\left(\frac{1}{\alpha T} + \alpha^\gamma\right)$. Consequently, it remains to control $\E\|\hat\beta_m - \hat\beta\|^2$. To that end, note that if $\hat\psi_m$ solves
	\begin{equation*}
	(\alpha  I + \hat K_m\hat K^*)\hat\psi_m = \hat r,
	\end{equation*}
	then $\hat\beta_m = \hat K^*\hat\psi_m$. Therefore,
	\begin{equation*}
	\hat\beta_m = \hat K^*(\alpha I + \hat K_m\hat K^*)^{-1}\hat r.
	\end{equation*}
	Next, decompose
	\begin{equation*}
	\begin{aligned}
	\hat\beta_m - \hat \beta & = \hat K^*(\alpha I + \hat K_m\hat K^*)^{-1}\hat r - \hat K^*(\alpha I + \hat K\hat K^*)^{-1}\hat r \\
	& = \hat K^*\left[(\alpha I + \hat K_m\hat K^*)^{-1} - (\alpha I + \hat K\hat K^*)^{-1}\right]\hat r \\
	& = \hat K^*(\alpha I + \hat K\hat K^*)^{-1}\left[\hat K\hat K^* - \hat K_m \hat K^*\right](\alpha I + \hat K_m\hat K^*)^{-1}\hat r \\
	& = \hat K^*(\alpha I + \hat K\hat K^*)^{-1}(\hat K - \hat K_m)\hat K^*(\alpha I + \hat K_m\hat K^*)^{-1}\hat r.
	\end{aligned}
	\end{equation*}
	Then
	\begin{equation*}
	\begin{aligned}
	\left\|\hat\beta_m - \hat\beta\right\|^2 & \leq \left\|\hat K^*(\alpha I + \hat K\hat K^*)^{-1}\right\|^2_\infty\left\|(\hat K - \hat K_m)\hat K^*\right\|^2_\infty\left\|(\alpha I + \hat K_m\hat K^*)^{-1}\right\|^2_\infty\|\hat r\|^2 \\
	& \leq \frac{\|\hat r\|}{4\alpha^3}\left\|(\hat K - \hat K_m)\hat K^*\right\|_\infty^2.
	\end{aligned}
	\end{equation*}
	Next, the expression inside of the operator norm is the integral operator on $L_2$
	\begin{equation*}
	(\hat K - \hat K_m)\hat K^*\psi = \int\psi(u)\left(\int \hat k(s,v)\hat k(s,u)\dx s - \sum_{j=1}^m\hat k(s_j,v)\hat k(s_j,u)\delta_j\right)\dx u.
	\end{equation*}
	Therefore, by the same computations as in Eq.~\ref{eq:hsboudn} and the triangle inequality
	\begin{equation*}\small
	\begin{aligned}
	\left\|(\hat K - \hat K_m)\hat K^*\right\|_\infty & \leq \left\|\int\hat k(s,.)\hat k(s,.)\dx s - \sum_{j=1}^m\hat k(s_j,.)\hat k(s_j,.)\delta_j\right\| \\
	& = \left\|\frac{1}{T^2}\sum_{t=1}^T\sum_{k=1}^T\Psi(.,W_t)\Psi(.,W_k)\left\{\int Z_t(s)Z_k(s)\dx s - \sum_{j=1}^mZ_t(s_j)Z_k(s_j)\delta_j\right\} \right\| \\
	& \leq \frac{1}{T^2}\sum_{t=1}^T\sum_{k=1}^T\|\Psi(.,W_t)\|\|\Psi(.,W_k)\|\left|\int Z_t(s)Z_k(s)\dx s - \sum_{j=1}^mZ_t(s_j)Z_k(s_j)\delta_j\right| \\
	& \leq \max_{1\leq t\leq T}\|\Psi(.,W_t)\|^2\max_{1\leq t,k\leq T}\left|\sum_{j=1}^mZ_t(s_j)Z_k(s_j)\delta_j - \int Z_t(s)Z_k(s)\dx s\right|.
	\end{aligned}
	\end{equation*}
	Under Assumption~\ref{as:discretization} (i)
	\begin{equation*}
	\begin{aligned}
	\left|Z_t(s_j)Z_k(s_j) - Z_t(s)Z_k(s)\right| & \leq |Z_t(s_j) - Z_t(s)||Z_k(s_j)| + |Z_t(s)||Z_k(s_j) - Z_k(s)| \\
	& \leq 2L^2|s_j - s|^\kappa,
	\end{aligned}
	\end{equation*}
	and whence
	\begin{equation*}
	\begin{aligned}
	\left|\sum_{j=1}^mZ_t(s_j)Z_k(s_j)\delta_j - \int Z_t(s)Z_k(s)\dx s\right| & = \left|\sum_{j=1}^m\int_{s_{j-1}}^{s_j}\left\{Z_t(s_j)Z_k(s_j) - Z_t(s)Z_k(s)\right\}\dx s\right| \\
	& \leq \sum_{j=1}^m\int_{s_{j-1}}^{s_j}\left|Z_t(s_j)Z_k(s_j) - Z_t(s)Z_{k}(s)\right|\dx s \\
	& = 2L^2\sum_{j=1}^m\int_{s_{j-1}}^{s_j}|s_j - s|^\kappa\dx s \\
	& \leq 2L^2\max_{1\leq j\leq m}\delta_j^\kappa.
	\end{aligned}
	\end{equation*}
	Therefore,
	\begin{equation*}
	\begin{aligned}
		\E\left\|\hat\beta_m - \hat\beta\right\|^2 & \leq \frac{\E\|\hat r\|}{\alpha^3}L^4\bar\Psi^2\max_{1\leq j\leq m}\delta_j^{2\kappa} \\
		& = O\left(\frac{\Delta_m^{2\kappa}}{\alpha^3}\right) \\
		& = O\left(\frac{1}{\alpha T}\right),
	\end{aligned}
	\end{equation*}
	where the second line follows under Assumption~\ref{as:data_ts} by Lemma~\ref{lemma:parametric_rate_ts} and the last under Assumption~\ref{as:discretization} (iii).
\end{proof}

\end{document}